%% file: arxiv_paper.tex
\documentclass[12pt,english]{amsart}

\usepackage[T1]{fontenc}
\usepackage[utf8]{luainputenc}

\usepackage{hyperref}

\usepackage[letterpaper]{geometry}
\usepackage{silence}
\usepackage{amsthm}
\usepackage{cleveref}
\usepackage{amstext}
\usepackage{setspace}
\usepackage{multirow}           

\usepackage{babel}

\usepackage{algorithm}
\usepackage{algpseudocode}

\makeatletter
\def\listofalgorithms{\@starttoc{loa}\listalgorithmname}
\def\l@algorithm{\@tocline{0}{3pt plus2pt}{0pt}{1.9em}{}}
\makeatother

\makeatletter
\newcommand*{\algrule}[1][\algorithmicindent]{\makebox[#1][l]{\hspace*{.5em}\vrule height .75\baselineskip depth .25\baselineskip}}%

\newcount\ALG@printindent@tempcnta
\def\ALG@printindent{%
    \ifnum \theALG@nested>0
    \ifx\ALG@text\ALG@x@notext
    \addvspace{-3pt}
    \else
    \unskip
    \ALG@printindent@tempcnta=1
    \loop
    \algrule[\csname ALG@ind@\the\ALG@printindent@tempcnta\endcsname]%
    \advance \ALG@printindent@tempcnta 1
    \ifnum \ALG@printindent@tempcnta<\numexpr\theALG@nested+1\relax
    \repeat
    \fi
    \fi
}%
\usepackage{etoolbox}
\patchcmd{\ALG@doentity}{\noindent\hskip\ALG@tlm}{\ALG@printindent}{}{\errmessage{failed to patch}}
\makeatother
\makeatother

\newcommand{\FF}{\mathbb{F}}
\newcommand{\tr}{\top}

\newcommand{\elementfun}[2]{{#1}^{({#2})}}
\renewcommand{\P}[1]{\elementfun{P}{#1}} 
\newcommand{\e}[1]{\elementfun{e}{#1}}

\algnewcommand{\True}{\textbf{true}\space}
\algnewcommand{\False}{\textbf{false}\space}
\algnewcommand{\And}{\textbf{and}}
\algnewcommand{\Or}{\textbf{ or }}
\algnewcommand{\Break}{\textbf{break}}%
\algnewcommand{\Continue}{\textbf{continue}}%
\algnewcommand{\algorithmicgoto}{\textbf{go to}}%
\algnewcommand{\Goto}[1]{\algorithmicgoto~\ref{#1}}%

\DeclareMathOperator{\rank}{rank} 
\DeclareMathOperator{\wt}{weight} 
\DeclareMathOperator{\image}{im}  

\theoremstyle{theorem}
\newtheorem{theorem}{Theorem}
\newtheorem{lemma}{Lemma}
\newtheorem{proposition}{Proposition}

\theoremstyle{definition}
\newtheorem{definition}{Definition}

\theoremstyle{remark}
\newtheorem{remark}{Remark}
\newtheorem{example}{Example}

\title[Beyond RAID~6]{Beyond RAID~6 --- an efficient systematic code
  protecting against multiple errors, erasures, and silent data
  corruption}

\author{Mohamad Moussa}
\address{Department of Mathematics, University of Arizona, Tucson, AZ 85721, USA}
\email{moussa7@math.arizona.edu}

\author{Marek Rychlik}
\address{Department of Mathematics, University of Arizona, Tucson, AZ 85721, USA}
\email{rychlik@email.arizona.edu}

\newcommand{\TrademarkNotice}{%
  PentaRAID\texttrademark is the trademark used by \emph{Xoralgo
    Inc.}, a company formed by the authors in collaboration with the
  University of Arizona, to pursue commercial implementations of the
  technology based on the research described in the current paper.}



\begin{document}

\begin{abstract}
  We describe a replacement for RAID~6, based on a new linear,
  systematic code, which detects and corrects any combination of $E$
  errors (unknown location) and $Z$ erasures (known location) provided
  that $Z+2E \leq 4$.  We investigate some scenarios for error
  correction beyond the code's minimum distance, using list
  decoding. We describe a decoding algorithm with quasi-logarithmic
  time complexity, when parallel processing is used:
  $\approx O(\log N)$ where $N$ is the number of disks in the array
  (similar to RAID~6).
  
  By comparison, the error correcting code implemented by RAID~6
  allows error detection and correction only when $(E,Z)=(1,0)$,
  $(0,1)$, or $(0,2)$. Hence, when in degraded mode (i.e., when
  $Z \geq 1$), RAID~6 loses its ability for detecting and correcting
  random errors (i.e., $E=0$), leading to data loss known as silent
  data corruption. In contrast, the proposed code does not experience
  silent data corruption unless $Z \geq 3$.
  
  The aforementioned properties of our code, the relative simplicity
  of implementation, vastly improved data protection, and low
  computational complexity of the decoding algorithm, make our code a
  natural successor to RAID~6.  As this code is based on the use of
  quintuple parity, this justifies the name PentaRAID\texttrademark\
  for the RAID technology implementing the ideas of the current paper.
\end{abstract}

\maketitle

\setcounter{tocdepth}{1}
\tableofcontents


\section{Introduction}
RAID (Redundant Arrays of Inexpensive Drives) was introduced as a
method of increasing reliability of data storage systems
\cite{Patterson_1988}. The original data is stored on $k$ devices, for
simplicity called ``drives''. Redundancy is added by using extra $p$
drives storing summary information (``parity'') computed from the
original data.  Parity information is used to reconstruct data lost
due to failure of some of the drives. The practicality of this
approach depends on the existence of efficient algorithms for
computing parity and reconstructing lost data. There is a trade-off
involving efficiency, protection from data loss, and cost. As RAID
evolved, several specific schemes were proposed, known as RAID levels.
RAID~6 is one popular scheme, in which $p=2$ and $k$ is arbitrary
\cite{WikipediaRaidLevels,Anvin_2009}. RAID~6 is limited in its
ability to protect data by $p=2$, which makes it possible to recover
from erasures $Z$ and errors $E$ provided that $Z+2E \leq 2$.  The
need for RAID~6 replacement has been understood for quite a
while~\cite{Plank_RS_Tutorial, Plank_Correction}; it has even been
predicted that RAID~6 will cease to work in year approximately 2019,
due to evolving storage capacity and application needs
~\cite{Leventhal:2009:TRB:1661785.1670144}. As $k$ increases, the
probability of having 2 failed drives simultaneously increases. It
should be noted that recovery takes a significant amount of time with
current large drives (say, $>1$ day for modest values of $k$). There
exist codes which allow arbitrary $p$, such as Reed-Solomon codes, but
they come with significant computational overhead and added complexity
of implementation. In the current paper we introduce a code which uses
$p=5$, allows recovery from up to 2 drive failures at unknown
locations and up to 4 drives at known locations, and which is nearly
as easy to implement as the RAID~6 code. Moreover, the knowledge
required to implement it is similar, and it amounts to familiarity
with Galois field arithmetic in the scope of the popular Anvin's paper
\cite{Anvin_2009}. Thus, we hope that PentaRAID\texttrademark will
successfully fill the gap between RAID~6 and Reed-Solomon
codes\footnote{\TrademarkNotice}
It
should be noted that our algorithm admits an implementation which uses
a constant, very small number of Galois field operations per error,
independent of the size of the underlying Galois field, if the
arithmetical operations in the field have constant time (e.g.\ use
lookup tables).  Thus the algorithm has better computational
complexity than the alternatives, and easily supports a large number
of drives, e.g. $254=2^8-2$ if the underlying Galois field is
$GF(256)$, and $65,533=2^{16}-2$ if the field is $GF(2^{16})$
\cite{WikipediaRaidLevels,Anvin_2009}.

The remainder of the paper is organized as follows. Section
\ref{sec:notations} introduces notations and some preliminary
considerations.  In Section \ref{sec:main-section}, we prove the main
theoretical result of the current paper,
Theorem~\ref{thm:injectivity-of-error-to-syndrome-mapping}.  A special
case in which $E=0$ (i.e erasure code) is discussed in Section
\ref{sec:erasures}. An efficient way of solving a quadratic equations
over a field of characteristic $2$ is presented in Section
\ref{sec:quadratic-equation}. Section \ref{sec:decoding-algorithm}
describes in detail the decoding algorithm for the proposed code. In
sections \ref{sec:two_unknown}, \ref{sec:complexity}, we present the
decoding technique for the case of dual disk corruption (i.e $E=2$ and
$Z=0$), and its computational complexity, respectively. In the last
three sections \ref{sec:three_drives},~\ref{sec:degraded_mode}, and
\ref{sec:four_drives} we investigate the error correction capabilities
beyond the minimum distance.

\section{Preliminary Considerations and Notations}
\label{sec:notations}
The bulk of the algebraic operations in the current paper are over the
finite field $\FF$ of characteristic two, and thus addition is the
same as subtraction. In addition to that, we will use frequently the
Frobenius identity ${(a+b)}^{2}=a^{2}+b^{2}$. In particular, we can
apply our results to the most commonly used finite field: the Galois
field $GF(256)$. There are minor differences when $\FF=GF(2^m)$, where
$m\ge 1$ is otherwise an arbitrary integer, and the results 
of the paper are applicable to all these fields.

The main object studied in this paper is a linear, systematic code
over the Galois field $\FF$. As it is customary, we define the code by
its generator and parity matrices. Consider the generator matrix
\begin{equation}
  \label{eqn:generator-matrix}
  G=\left(
  \begin{array}{c}
    I_{k\times k} \\
    \hline
    P_{5\times k}
  \end{array}
  \right),
\end{equation} 
and the parity check matrix 
\begin{equation}
\label{eqn:parity-check-matrix}
  H=\left(
  \begin{array}{c|c}
    P_{5\times k} & I_{5\times 5}
  \end{array}
  \right),
\end{equation}
where 
\begin{equation}
  \label{eqn:parity-matrix}
  P_{5\times k}=\left(
  \begin{array}{cccc}
    1                        &  1                         & \cdots  & 1\\
    \alpha_{1}                & \alpha_{2}                 & \cdots  & \alpha_{k}\\
    \alpha_{1}^{2}            & \alpha_{2}^{2}              & \cdots  & \alpha_{k}^{2}\\
    \alpha_{1}^{3}            & \alpha_{2}^{3}              & \cdots  & \alpha_{k}^{3}\\
    \alpha_{1}^{2}+\alpha_{1}  & \alpha_{2}^{2}+\alpha_{2}   & \cdots  & \alpha_{k}^{2}+\alpha_{k}
  \end{array}
  \right).
\end{equation}
As we will often refer to the structure of the column of this matrix, it will be convenient
to introduce a vector depending on $\rho\in\FF$:
\begin{equation}
  \label{eqn:P-rho}
  \P{\rho}=\left(\begin{array}{ccccc}
    1&\rho&\rho^2&\rho^3&\rho(\rho+1)
  \end{array}\right)^\tr.
\end{equation}
Using this notation, we can define the columns of $P=P_{5\times k}$ to be
$P_i=\P{\alpha_i}$.

The entries $\alpha_{i}$, $i=1,\ldots,k$, are distinct elements of
$\FF$, excluding the zero element, and any one of the $3^{rd}$ roots of
unity distinct from $1$ (if such a root exists in $\FF$). So, we have the 
following bound:
\begin{equation*}
  k \leq k_{\max}(\FF)=\begin{cases}
    |\FF|-1 &\text{if there is no $3^{rd}$ roots of unity in $\FF$,}\\
    |\FF|-2 &\text{if there is a $3^{rd}$ root of unity in $\FF$}.
  \end{cases}
\end{equation*}
For example, if $\FF=GF(256)=GF(2^8)$, the equation
$\zeta^3-1=(\zeta-1)\,(\zeta^2+\zeta+1)$ has two solutions in $\FF$
distinct from $1$, and thus $k\leq k_{\max}=254$.  The reason behind the
exclusion of a $3^{rd}$ root of unity will be clarified later.

It should be noted that, in contrast with most linear codes considered
in literature, the rows of the matrix $P_{5\times k}$ are not linearly
independent: the sum of rows $2$ and $3$ is equal to row $5$. Also, by
dropping the $5^{th}$ row, we obtain a $4\times 4$ Vandermonde matrix.

For $k$ \emph{data} disks $D_{1},\ldots,D_{k}$ where $k\leq k_{\max}$, let
\begin{equation}
  m=\left(\begin{array}{cccc}
    D_{1} &
    D_{2} &
    \ldots &
    D_{k}
  \end{array}
  \right)^\tr%
\end{equation}
be the original message (original values of the data drives), and
\begin{equation}
  t=G\cdot m=\left(
  \begin{array}{cccc|ccccc}
    D_{1} & 
    D_{2} &
    \ldots & 
    D_{k}&
    P_{1} &
    P_{2} &
    P_{3} &
    P_{4} &
    P_{5}
  \end{array}\right)^\tr
\end{equation}
be the transmitted message. Note that $H\cdot G$ is a zero matrix,
since the field is of characteristic $2$, and thus $H\cdot t=H\cdot G\cdot m=0$.

Consider the received message $r$ to be the transmitted message, plus
an error message $r=t+e$, where $e$ is an element of the vector space
$\FF^{k+5}$. As it is customary, we think of $e$ as random, thus to be
modeled with probability theory.

The \emph{syndrome} of a received message is 
\begin{equation*}
  s=H\cdot r=H\left(t+e\right)=H\cdot t+H\cdot e=0+H\cdot e=H\cdot e.
\end{equation*}
In short, the syndrome of a message is $s=H\cdot e$. It is a column
vector of size $5$ by $1$.

In this paper, we prove that the linear code defined by the above pair of
matrices $G$ and $H$ can be uniquely decoded using an approach widely
known as \emph{syndrome decoding}.  We will provide an algorithm which
will perform the decoding task. As the first step, we prove the following theorem:

\begin{theorem}[Injectivity of error-to-syndrome mapping]
\label{thm:injectivity-of-error-to-syndrome-mapping}
For one-error and two-error patterns, the function (the parity check
matrix $H$) that maps the error $e$ into the the syndrome
$s=H\cdot e$, is an injective transformation.  Thus, having the
syndrome $s$, can tell us the locations of the errors and the error
values.
\end{theorem}

The proof of Theorem~\ref{thm:injectivity-of-error-to-syndrome-mapping}
will occupy the entire next section, as the proof is divided into 
a significant number of distinct cases which require detailed analysis.

The following theorem describes our code in classical terms:
\begin{theorem}[On dimension, length and distance of the code]
  \label{thm:dimension-length-and-distance}%
  The linear code given by generator matrix $G$ and the parity check matrix
  $H$ is a systematic linear code with length $k+5$, rank (or
  dimension) $k$ and distance $5$.
\end{theorem}
\begin{proof}
  We defer the proof of the statement concerning the distance to
  Proposition~\ref{thm:distance-of-the-code-is-5}. The length of the
  code is clearly $k+5$.  The rank of the code is $k$ as the columns
  of the generator matrix are linearly independent. It is clear that
  $\rank G\leq k$, the number of columns.  Erasing parity matrix
  $P_{5\times k}$ from the generator matrix leaves us with the
  identity matrix $I_{k\times~k}$, so $\rank G \geq k$. Combining
  these estimates, $\rank G=k$.
\end{proof}

In some applications the locations of the failed (or erased) drives
are known. The following theorem is an easy consequence of
Theorem~\ref{thm:dimension-length-and-distance}. We address this case
in Section~\ref{sec:erasures}, in particular, in
Theorem~\ref{thm:correctability-of-four-erasures}, proving that up to
$4$ erasures at known locations can be corrected.

In general, it is advantageous for a linear, forward error correcting
code (FEC) to have distance as large as possible. An $[N,K,D]$ linear
FEC, where $N$ is the length, $K$ is the rank and
$D$ is the distance, can correct $E$ errors and $Z$ (known) erasures
if the condition $Z+2\times E \leq D - 1$ is satisfied
(\cite{MoonLinearBlockCodes}, pp.~104--105). As an example, a liner
FEC with distance $D=5$ can correct a combination of $Z=2$ (known)
erasures and $E=1$ error (at unknown location).

\section{Proofs of the Main Results}
\label{sec:main-section}
In this section we prove the main theoretical result of the current
paper, Theorem~\ref{thm:injectivity-of-error-to-syndrome-mapping}.
Let us begin with stating useful results from coding theory.
\begin{definition}[Weight of a vector]
  The weight of a vector $e\in\FF^n$ is the number of non-zero entries:
  \begin{equation*}
    \wt e = \# \{i\in\{1,2,\ldots,n\}\,:\, e_i \neq 0 \}.
  \end{equation*}
\end{definition}
Let $\mathcal{S}^n_{\leq \nu} \subseteq \FF^n$ be the set of all vectors
$e$ with at most $\nu$ non-zero entries, i.e.\ whose weight is at most
$\nu$. In order to prove
Theorem~\ref{thm:injectivity-of-error-to-syndrome-mapping}, we will
prove that the function
\begin{equation}
  \FF^{k+5}\supset \mathcal{S}^{k+5}_{\leq 2} \ni e \mapsto H\cdot e \in \FF^5 
\end{equation}
is injective. 

Let us recall a known theorem from the theory of linear
codes~\cite{MoonLinearBlockCodes}, p.~88, Theorem~3.3:
\begin{theorem}
  For a linear code $C\subseteq\FF^n$ with a parity check matrix $H$,
  the minimum distance is $d$ if and only if both of the following
  conditions hold:
  \begin{enumerate}
  \item every set of $d-1$ columns of $H$ is linearly independent;
  \item some set of $d$ columns of $H$ is linearly dependent.
  \end{enumerate}
\end{theorem}
We also have the theorem stating the connection between the distance
and capability to detect and correct errors at unknown locations
\cite{MoonLinearBlockCodes},  p.~101.
\begin{theorem}
  Let $C\subseteq\FF^{k+r}$ be a linear code with distance $d$.  Then the
  mapping
  \begin{equation}
    \FF^{k+r}\supset \mathcal{S}^{k+r}_{\leq \nu} \ni e \mapsto H\cdot e \in \FF^r 
  \end{equation}
  is injective if 
  \[ \nu \leq \left\lfloor \frac{d-1}{2} \right\rfloor. \]
\end{theorem}
\begin{proposition}\label{thm:distance-of-the-code-is-5}%
  Let $C\subseteq\FF^{k+5}$ be the code defined by the generator
  matrix $G$ and parity check matrix $H$, given by
  equations~\eqref{eqn:generator-matrix},~\eqref{eqn:parity-check-matrix}
  and~\eqref{eqn:parity-matrix}. Then the minimum distance of $C$ is
  $5$.
\end{proposition}
\begin{proof}
  First, let us show that any combination of $4$ distinct columns of $H$ is linearly independent.
  Any $4$ columns of $H$ are obtained by choosing $r\leq 4$ columns of the identity matrix $I_{5\times 5}$,
  and then choosing $t=4-r$ columns of $P_{5\times k}$, i.e.\ vectors of the form
  \begin{equation*}
    P_{i}^\tr=
    \left(\begin{array}{ccccc}
      1 & \alpha_i & \alpha_i^2 &\alpha_i^3 & \alpha_i(\alpha_i+1)
    \end{array}\right)^\tr.
  \end{equation*}
  To explain the method of proof, let us first consider a special case, when
  2 columns of $I_5$ are selected, 
  \begin{align*}
    I_{2}^\tr&=\left(\begin{array}{ccccc}
      0 &1 &0 &0 &0
    \end{array}\right)^\tr,\\
    I_{4}^\tr&=\left(\begin{array}{ccccc}
      0 &0 &0 &1 &0\end{array}
      \right)^\tr.
  \end{align*}
  Let us consider the matrix $A$ which contains the vectors whose
  independence we study.  It is our prerogative to order the vectors
  in any order, and we choose to write the ones coming from the
  identity matrix before the ones coming from the parity matrix.
  Also, we choose to perform column rather than row reduction to find
  the echelon form.
  \begin{equation*}
    A=\left(\begin{array}{cccc}
      I_{2} & I_{4}  & P_{i} &  P_{j}
    \end{array}\right)
    =\left(\begin{array}{cc|cc}
      0 & 0 & 1 & 1 \\
      1 & 0 & \alpha_i & \alpha_j\\
      0 & 0 & \alpha_i^2 & \alpha_j^2\\
      0 & 1 & \alpha_i^3 & \alpha_j^3\\
      0 & 0 &  \alpha_i(\alpha_i+1) &  \alpha_j(\alpha_j+1)
    \end{array}\right).
  \end{equation*}
  We need to prove that this matrix has rank $4$. As rank is invariant
  under elementary column operations, we can subtract $\alpha_i\times column_1$ 
  form column $3$, and $\alpha_j\times column_1$ from column $4$. Similarly,
  we can subtract multiples of $column_{2}$ from $column_3$ and $column_4$ to eliminate
  entries in row $4$. The resulting matrix is:
  \begin{equation*}
    A^{(1)} = \left(\begin{array}{cc|cc}
      0 & 0 & 1 & 1 \\
      1 & 0 & 0 & 0 \\
      0 & 0 & \alpha_i^2 & \alpha_j^2\\
      0 & 1 & 0 & 0 \\
      0 & 0 &  \alpha_i(\alpha_i+1) &  \alpha_j(\alpha_j+1)
    \end{array}\right).
  \end{equation*}
  We can rearrange the rows of $A^{(1)}$ to bring it to the column echelon form:
  \begin{equation*}
    A^{(2)}=\left(\begin{array}{c|c}
      I_{2\times 2} & 0\\
      \hline
      0 & B
    \end{array}\right),\quad
    B=\left(\begin{array}{cc}
      1 & 1 \\
      \alpha_i^2 & \alpha_j^2\\
      \alpha_i(\alpha_i+1) &  \alpha_j(\alpha_j+1)
    \end{array}\right).
  \end{equation*}
  It is easy to see that $\rank A^{(2)}=\rank I_{2\times 2} + \rank B=2+\rank B$.
  Therefore we need to show that $\rank B=2$.

  We observe that matrix $B$ is obtained by erasing $2$ rows of the matrix
  $C$ given by:\label{page:C-matrix}
  \begin{equation*}
    C=\left(\begin{array}{ccccc}
      1 & 1 \\
      \alpha_i & \alpha_j\\
      \alpha_i^2 & \alpha_j^2\\
      \alpha_i^3 & \alpha_j^3\\
      \alpha_i(\alpha_i+1) &  \alpha_j(\alpha_j+1)
    \end{array}\right).
  \end{equation*}
  The general case is also reduced to calculating rank of a specific matrix $B$.
  We consider the matrix
  \begin{equation*}
    C = \left(\begin{array}{cccc}
      P_{i_1} & P_{i_2} & \ldots P_{i_t}
    \end{array}\right)
  \end{equation*}
  obtained from $P_{5\times k}$ by taking any subset of $t\leq 4$ columns,
  and proving that each of its submatrices obtained by erasing $r=4-t$ rows
  has rank $t=4-r$. As $r$ corresponds to the number of columns taken from the
  identity matrix $I_{5\times 5}$, the total rank is $r+t = (4-t)+t=4$ as claimed.

  Let us consider various cases as $t=0,1,2,3,4$.

  If $t=0$ then the matrix has $0$ columns, and the rank is $0$.

  If $t=1$ then $r=3$ entries are erased from a single-column matrix $C=P_{i}$.
  At least one of the two remaining entries is non-zero so $\rank C=1=t$, as required.

  If $t=2$ then $r=2$. The matrix $C^{lm}$, obtained by erasing two rows $l<m$, from $C$, is thus $3\times 2$, and 
  is in one of the ${5\choose 2}=10$ forms:\label{page:C-matrices}
  \begin{align*}
    &C^{45}=\left(\begin{array}{cc}
      1        & 1          \\
      \alpha_i & \alpha_j   \\
      \alpha_i^2 & \alpha_j^2
    \end{array}\right),
    \quad
    C^{35}=\left(\begin{array}{cc}
      1        & 1          \\
      \alpha_i & \alpha_j   \\
      \alpha_i^3 & \alpha_j^3
    \end{array}\right),
    \quad
    C^{34}=\left(\begin{array}{cc}
      1        & 1        \\
      \alpha_i & \alpha_j \\
      \alpha_i(\alpha_i+1) & \alpha_j(\alpha_j+1)
    \end{array}\right),\\
    &
    C^{25}=\left(\begin{array}{cc}
      1        & 1          \\
      \alpha_i^2 & \alpha_j^2   \\
      \alpha_i^3 & \alpha_j^3
    \end{array}\right),
    \quad
    C^{24}=\left(\begin{array}{cc}
      1        & 1        \\
      \alpha_i^2 & \alpha_j^2 \\
      \alpha_i(\alpha_i+1) & \alpha_j(\alpha_j+1)
    \end{array}\right),
    \quad
    C^{23}=\left(\begin{array}{cc}
      1        & 1        \\
      \alpha_i^3 & \alpha_j^3 \\
      \alpha_i(\alpha_i+1) & \alpha_j(\alpha_j+1)
    \end{array}\right),\\
    &
    C^{15}=\left(\begin{array}{cc}
      \alpha_i & \alpha_j    \\
      \alpha_i^2 & \alpha_j^2 \\
      \alpha_i^3 & \alpha_j^3
    \end{array}\right),
    \quad
    C^{14}=\left(\begin{array}{cc}
      \alpha_i & \alpha_j \\
      \alpha_i^2 & \alpha_j^2 \\
      \alpha_i(\alpha_i+1) & \alpha_j(\alpha_j+1)
    \end{array}\right),
    \quad
    C^{13}=\left(\begin{array}{cc}
      \alpha_i & \alpha_j \\
      \alpha_i^3 & \alpha_j^3 \\
      \alpha_i(\alpha_i+1) & \alpha_j(\alpha_j+1)
    \end{array}\right),\\
    &C^{12}=\left(\begin{array}{cc}
      \alpha_i^2 & \alpha_j^2 \\
      \alpha_i^3 & \alpha_j^3 \\
      \alpha_i(\alpha_i+1) & \alpha_j(\alpha_j+1)
    \end{array}\right).
  \end{align*}
  In order for each of the matrices to have the required rank $2$, every
  one of the above matrices should have a $2\times 2$ non-singular
  submatrix.

  First 3 matrices, $C^{45}$, $C^{35}$ and $C^{34}$, contain the matrix
  \begin{equation*}
    \left(\begin{array}{cc}
      1        & 1          \\
      \alpha_i & \alpha_j   \\
    \end{array}\right)
  \end{equation*}
  which has determinant $\alpha_j-\alpha_i\neq 0$, in view of our
  assumption that all $\alpha_i$ are distinct.
  
  Matrices $C^{25}$ and $C^{24}$ contain the matrix
  \begin{equation*}
    \left(\begin{array}{cc}
      1        & 1          \\
      \alpha_i^2 & \alpha_j^2   \\
    \end{array}\right)
  \end{equation*}
  which has determinant $\alpha_j^2-\alpha_i^2=(\alpha_j-\alpha_i)^2$
  by the Frobenius identity. Hence, it is also $\neq 0$.
  
  Matrices $C^{15}$ and $C^{14}$ contain the matrix
  \begin{equation*}
    \left(\begin{array}{cc}
      \alpha_i        & \alpha_j          \\
      \alpha_i^2 & \alpha_j^2   \\
    \end{array}\right)
  \end{equation*}
  which has determinant $\alpha_i\alpha_j^2-\alpha_j\alpha_i^2=\alpha_i\alpha_j(\alpha_j-\alpha_i)\neq 0$, as
  $\alpha_i$ are all non-zero and distinct.

  Matrix $C^{13}$ contains the matrix
  \begin{equation*}
    \left(\begin{array}{cc}
      \alpha_i        & \alpha_j          \\
      \alpha_i^3 & \alpha_j^3   \\
    \end{array}\right)
  \end{equation*}
  which has determinant
  $\alpha_i\alpha_j^3-\alpha_j\alpha_i^3=\alpha_i\alpha_j(\alpha_j^2-\alpha_i^2)=\alpha_i\alpha_j(\alpha_j-\alpha_i)^2\neq
  0$ in view of $\alpha_i,\alpha_j\neq 0$ and $\alpha_i\neq\alpha_j$.

  Matrix $C^{12}$ contains the matrix
  \begin{equation*}
    \left(\begin{array}{cc}
      \alpha_i^2        & \alpha_j^2          \\
      \alpha_i^3 & \alpha_j^3   \\
    \end{array}\right)
  \end{equation*}
  which has determinant
  $\alpha_i^2\alpha_j^3-\alpha_j^2\alpha_i^3=\alpha_i^2\alpha_j^2(\alpha_j-\alpha_i)\neq0$
  in view of $\alpha_i,\alpha_j\neq 0$ and $\alpha_i\neq\alpha_j$.
  
  The only matrix left in this group of 10 is $C^{23}$. It has three $2\times2$ submatrices, some of which can have
  the determinant $0$. We claim that the following two submatrices of $C^{23}$,
  \begin{equation*}
    \left(\begin{array}{cc}
      1 & 1\\
      \alpha_{i}\left(\alpha_{i}+1\right) & \alpha_{j}\left(\alpha_{j}+1\right)
    \end{array}\right),
    \quad\text{and}\quad
    \left(\begin{array}{cc}
      \alpha_{i}^{3} & \alpha_{j}^{3}\\
      \alpha_{i}\left(\alpha_{i}+1\right) & \alpha_{j}\left(\alpha_{j}+1\right)
    \end{array}\right),
  \end{equation*}
  cannot be both singular. Indeed, their determinants are:
  \begin{equation*}
    \alpha_{j}\left(\alpha_{j}+1\right)+\alpha_{i}\left(\alpha_{i}+1\right),
    \qquad
    \alpha_{i}^3\alpha_j\left(\alpha_{j}+1\right)+\alpha_{j}^3\alpha_i\left(\alpha_{i}+1\right).
  \end{equation*}
  Let us suppose that both determinants are $0$. We then have (after dividing the second one by $\alpha_i\alpha_j\neq 0$):
  \begin{equation*}
    \left\{\begin{array}{ccc}
    \alpha_{j}\left(\alpha_{j}+1\right)&=&\alpha_{i}\left(\alpha_{i}+1\right)\\
    \alpha_{i}^2\left(\alpha_{j}+1\right)&=&\alpha_{j}^2\left(\alpha_{i}+1\right).
    \end{array}\right.
  \end{equation*}
  The first equation and Frobenius identity imply:
  \begin{equation*}
    0=\alpha_i^2+\alpha_j^2+\alpha_i+\alpha_j = (\alpha_i+\alpha_j)^2+(\alpha_i+\alpha_j).
  \end{equation*}
  Thus, $\alpha_i+\alpha_j=0$ or $\alpha_i+\alpha_j=1$. Since $\alpha_i\neq \alpha_j$,
  $\alpha_i+\alpha_j\neq 0$. Thus $\alpha_i+\alpha_j=1$.

  If $\alpha_i=1$ then $0=\alpha_j(\alpha_j+1)$. Hence $\alpha_j=1$
  ($\alpha_j\neq 0$ by assumption).  But then $\alpha_i=\alpha_j=1$
  which contradicts the assumption that $\alpha$'s are distinct. Hence,
  we can divide equations side by side and obtain
  \begin{equation*}
    \frac{\alpha_j}{\alpha_i^2}=\frac{\alpha_i}{\alpha_j^2},
  \end{equation*}
  or $\alpha_i^3=\alpha_j^3$. Therefore,
  \begin{align*}
    1&=(\alpha_i+\alpha_j)^3=\alpha_i^3+3\alpha_i^2\alpha_j+3\alpha_i\alpha_j^2+\alpha_j^3\\
    &=\alpha_i^3+\alpha_i^2\alpha_j+\alpha_i\alpha_j^2+\alpha_j^3=\alpha_i\alpha_j(\alpha_i+\alpha_j) = \alpha_i\alpha_j.
  \end{align*}
  Hence $\alpha_i\alpha_j=1$ and $\alpha_i+\alpha_j=1$, which implies
  that $\alpha_i$ and $\alpha_j$ are the distinct roots of the equation
  $\alpha^2+\alpha+1=0$, i.e.\ are the distinct $3^{rd}$ roots of unity
  in $\FF$.  But we excluded one of the $3^{rd}$ roots of unity from
  amongst $\alpha_j$, $j=1,2,\ldots,k$, so we obtained a contradiction. Thus at least one of the matrices is invertible.

  If $t=3$ then $r=4-t=1$. In this case, matrix
  $C$ is given by:
  \begin{equation*}
    C=\left(\begin{array}{ccc}
      1 & 1 & 1\\
      \alpha_i & \alpha_j & \alpha_l\\
      \alpha_i^2 & \alpha_j^2 & \alpha_l^2\\
      \alpha_i^3 & \alpha_j^3 & \alpha_l^3\\
      \alpha_i(\alpha_i+1) &  \alpha_j(\alpha_j+1) &  \alpha_l(\alpha_l+1)
    \end{array}\right).
  \end{equation*}
  The submatrices obtained by erasing a single row from $C$ are:
  \begin{align*}
    &C^{5}=\left(\begin{array}{ccc}
      1 & 1 & 1\\
      \alpha_i & \alpha_j & \alpha_l\\
      \alpha_i^2 & \alpha_j^2 & \alpha_l^2\\
      \alpha_i^3 & \alpha_j^3 & \alpha_l^3\\
    \end{array}\right),
    \enskip
    C^{4}=\left(\begin{array}{ccc}
      1 & 1 & 1\\
      \alpha_i & \alpha_j & \alpha_l\\
      \alpha_i^2 & \alpha_j^2 & \alpha_l^2\\
      \alpha_i(\alpha_i+1) &  \alpha_j(\alpha_j+1) &  \alpha_l(\alpha_l+1)
    \end{array}\right),\enskip
    \\
    &C^{3}=\left(\begin{array}{ccc}
      1 & 1 & 1\\
      \alpha_i & \alpha_j & \alpha_l\\
      \alpha_i^3 & \alpha_j^3 & \alpha_l^3\\
      \alpha_i(\alpha_i+1) &  \alpha_j(\alpha_j+1) &  \alpha_l(\alpha_l+1)
    \end{array}\right),
    \enskip
    C^{2}=\left(\begin{array}{ccc}
      1 & 1 & 1\\
      \alpha_i^2 & \alpha_j^2 & \alpha_l^2\\
      \alpha_i^3 & \alpha_j^3 & \alpha_l^3\\
      \alpha_i(\alpha_i+1) &  \alpha_j(\alpha_j+1) &  \alpha_l(\alpha_l+1)
    \end{array}\right),
    \\
    &C^{1}=\left(\begin{array}{ccc}
      \alpha_i & \alpha_j & \alpha_l\\
      \alpha_i^2 & \alpha_j^2 & \alpha_l^2\\
      \alpha_i^3 & \alpha_j^3 & \alpha_l^3\\
      \alpha_i(\alpha_i+1) &  \alpha_j(\alpha_j+1) &  \alpha_l(\alpha_l+1)
    \end{array}\right).
  \end{align*}
  Of these 5 matrices, matrix $C^{5}$ and $C^{4}$ contain the $3\times3$ Vandermonde matrix:
  \begin{equation*}
    \left(\begin{array}{ccc}
      1 & 1 & 1\\
      \alpha_i & \alpha_j & \alpha_l\\
      \alpha_i^2 & \alpha_j^2 & \alpha_l^2\\
    \end{array}\right).
  \end{equation*}
  which has the determinant
  $(\alpha_j-\alpha_i)\,(\alpha_l-\alpha_i)\,(\alpha_l-\alpha_j) \neq 0$, in
  view of the fact that $\alpha_i$ are all distinct. Matrix $C^{1}$ of the 5
  contains a matrix related to Vandermonde,
  \begin{equation*}
    \left(\begin{array}{ccc}
      \alpha_i & \alpha_j & \alpha_l\\
      \alpha_i^2 & \alpha_j^2 & \alpha_l^2\\
      \alpha_i^3 & \alpha_j^3 & \alpha_l^3\\
    \end{array}\right)
  \end{equation*}
  which has the determinant
  $\alpha_i\alpha_j\alpha_l(\alpha_j-\alpha_i)\,(\alpha_l-\alpha_i)\,(\alpha_l-\alpha_j)
  \neq 0$.

  This leaves matrix $C^{3}$ and $C^{2}$. We find, using CAS, that
  \begin{equation*}
    \left|\begin{array}{ccc}
    \alpha_i & \alpha_j & \alpha_l\\
    \alpha_i^3 & \alpha_j^3 & \alpha_l^3\\
    \alpha_i(\alpha_i+1) &  \alpha_j(\alpha_j+1) &  \alpha_l(\alpha_l+1)
    \end{array}\right| = 
    \alpha_i\alpha_j\alpha_l(\alpha_i-\alpha_j)\,(\alpha_i-\alpha_l)\,(\alpha_j-\alpha_l)\neq 0,
  \end{equation*} 
  (this can also be seen by subtracting row 1 from row 3, and then swapping rows 2 and 3; the matrix
  becomes the same as $C^{1}$) and
  \begin{equation*}
    \left|\begin{array}{ccc}
    \alpha_i^2 & \alpha_j^2 & \alpha_l^2\\
    \alpha_i^3 & \alpha_j^3 & \alpha_l^3\\
    \alpha_i(\alpha_i+1) &  \alpha_j(\alpha_j+1) &  \alpha_l(\alpha_l+1)
    \end{array}\right| = 
    \alpha_i\alpha_j\alpha_l(\alpha_i-\alpha_j)\,(\alpha_i-\alpha_l)\,(\alpha_j-\alpha_l)\neq 0.
  \end{equation*}
  (this factorization is valid over a field of any characteristic; we
  can also use the following argument: subtract row 1 from row 3, then
  swap rows 1 and 3; the resulting matrix is now the same as matrix 5).  Since the
  matrix under the above determinant is a submatrix of both matrices $C^{3}$
  and $C^{2}$, they both have rank $3$.

  Finally, when $t=4$, we need to prove that $C$ itself has rank
  $4$. However, we observe that the first four rows of this matrix
  form a $4\times 4$ Vandermonde matrix, which is non-singular in view
  of $\alpha_i$ being distinct.  

  Clearly, any set of $5$ columns of the $P_{5\times k}$ matrix is
  linearly dependent because row $5$ is a sum of row $2$ and
  $3$. Hence, the proof is complete.
\end{proof}
As a corollary of the proof we obtain the following useful criterion:
\begin{theorem}[A criterion of a systematic code to have distance $d+1$]
  \label{thm:distance-of-the-code-is-5-criterion}
  A systematic code with parity matrix $P$ has distance $d+1$ iff
  every matrix $P'$ obtained from $P$ by taking $t\leq d$ columns of $P$
  and deleting $d-t$ rows of $P$ has rank $t$, i.e. $P'$ has a non-singular
  submatrix of size $t\times t$.
\end{theorem}
\section{Recovery from up to $4$ Erasures (known locations errors)}
\label{sec:erasures}
The following theorem addresses the situation when up to $4$ drives
\emph{at known locations} have been erased (or corrupted).
\begin{theorem}[Correctability of up to $4$ erasures]
\label{thm:correctability-of-four-erasures}
The code with generator matrix $G$ defined by \eqref{eqn:generator-matrix}
allows recovery from up to $4$ drive failures at known locations.
\end{theorem}
\begin{proof}
  The code is a systematic code with distance $d=5$.  The proof given
  below works for an arbitrary code with distance $d$ and parity check
  matrix $H$. Hence, any $d-1$ columns of $H$ are linearly
  independent.  If the locations of the erased drives are
  $i_1,i_2,\ldots,i_q $ then $e_i=0$ if
  $i\notin\{i_1,i_2,\ldots,i_q\}$.  Therefore
  \[ s= H\,e = \sum_{m=1}^q H_{i_m}e_{i_m} . \]
  Let $H'=\left(\begin{array}{c|c|c|c}H_{i_1}&H_{i_2}&\ldots&H_{i_q}\end{array}\right)$
  be the submatrix of $H$.
  If $q\leq d-1$ then there is a subset of $q\leq d-1$ rows of $H'$ such that the
  resulting matrix, which we will call $H''$, is non-singular.
  In view of 
  \[ H''\cdot e' = s' \]
  where $s'$ is obtained from $e$ and $s$ by selecting the same rows,
  and
  $e'=\left(\begin{array}{cccc}e_{i_1}&e_{i_2}&\ldots&e_{i_q}\end{array}\right)$,
  has a unique solution $e'$. Clearly, we can complete $e'$ to $e$ in
  a unique way, by setting
  \begin{equation*} 
    e_i=\begin{cases} 
    e'_m &\text{if $i=i_m$ for some $m$, $1\leq m\leq q$,}\\
    0 &\text{otherwise.}
    \end{cases}
  \end{equation*}
  The transmitted vector $t$ is found from $r=t+e$, yielding $t=r-e=r$.
  Note that $t=r+e$ if the field $\FF$ has characteristic $2$.
\end{proof}
\begin{remark}[On the choice of a non-singular submatrix]
  The choice of a specific $4\times 4$ non-singular submatrix in the
  proof of Theorem~\ref{thm:correctability-of-four-erasures} was
  further discussed, and made explicit, in the proof of
  Proposition~\ref{thm:distance-of-the-code-is-5}.
\end{remark}

\subsection{Recovery when positions of errors are known:}
Let us state a general principle which works when positions of all
errors are known, and so $E=0$.  Suppose that we have a
systematic code of distance $d + 1$ and we have a failure of $t$ data
disks at known locations (known erasures).  Let
$I\subseteq\{1,2,\ldots k\}$ be the set of data error locations.  In
addition, let $J$ be a subset of no more than $d-t$ known parity
error locations.  We know that a submatrix $P'$ of $P$ (the parity
matrix) obtained by taking only $t$ columns with indices in $I$ and
deleting $d-t$ rows of $P$ contains a $t\times t$ submatrix matrix
$P''$ which is non-singular (Theorem
\ref{thm:distance-of-the-code-is-5-criterion}). Thus we may choose
$P''$ whose row indices are not in $J$, i.e. do not correspond to
known parity errors.  Let the set of row indices be $K$, where $|K|=t$
and $K\cap J=\emptyset$.  We can use $P''$ to find the error values of
data errors. Then we may find the parity error values for all parities
$j\notin K$.  We simply use the equation:
\[ e_{k+j} = s_j - \sum_{i\in I} p_{j,i} e_i \]
These values, $e_{k+j}$, $j\notin K$ may or may not be $0$, which
determines the actual number of parity errors. In order to satisfy the
assumptions about parity errors, $e_{k+j}=0$ for $j\notin J$ is
required. Hence, when all errors are erasures, the decoding problem is
solved by linear algebra methods. As we will see, in other cases
non-linear methods are required, requiring sometimes delicate analysis
of systems of polynomial equations.

\section{Solving Quadratic Equations over $GF\left(2\right)$}
\label{sec:quadratic-equation}
The equation $a\,x^2+b\,x+c=0$ is solved by the quadratic formula
$$x_{1,2}=\frac{-b\pm\sqrt{b^2-4\,a\,c}}{2\,a}$$ over any field of
characteristic $\neq 2$. If the characteristic is $2$, the quadratic
formula obviously cannot work as the denominator is $0$. In the current
section we develop a replacement for the quadratic formula, which will allow us
to solve quadratic equations.
\begin{lemma}[On the difference of the roots of a quadratic equation]
  \label{thm:difference-of-roots-of-quadratic-equation}
  Let $a, b, c$ be constants in a Galois field $\FF$ of characteristic $2$
  and let $x$ be a variable. If $x_1\in\FF$ is
  a root of the quadratic equation
  \begin{equation}
    \label{eqn:quadratic-equation}
    a\,x^2+b\,x+c=0,\quad a\neq 0
  \end{equation}
  then the second root is given by the equation:
  \begin{equation*}
    x_2=x_1+\frac{b}{a}.
  \end{equation*}
\end{lemma}
\begin{proof}
  By direct calculation:
  \begin{align*}
    a\,x_2^2+b\,x_2+c&=a\left(x_1+\frac{b}{a}\right)^2+b(x_1+\frac{b}{a})+c\\
    &=a\left(x_1^2+\frac{b^2}{a^2}\right)+b\left(x_1+\frac{b}{a}\right)+c\\
    &=a\,x_1^2+\frac{b^2}{a}+b\,x_1+\frac{b^2}{a}+c\\
    &=\left(a\,x_1^2+b\,x_1+c\right)+2\frac{b^2}{a}=0.
  \end{align*}
\end{proof}
We shall focus on the algorithmic aspects of solving the quadratic
equation.  An obvious algorithm over a finite field is obtained by
trying all the elements until we find a root, i.e.\ perform a search.
The cost of this grows as the size of the field grows. In this section
we will develop an algorithm which uses constant time and linear
memory.

Our first observation is that if $b=0$ then the quadratic equation
becomes $a\,x^2+c=0$, or 
\begin{equation}
  \label{eqn:quadratic-no-linear-term}
  x^2=\frac{c}{a}.
\end{equation}
If $c=0$ then $x=0$.
if $c\neq 0$ then we need to compute the square root. The
Frobenius map $J:\FF\to\FF$ given by
\begin{equation}
  J(x)=x^2
\end{equation}
is an automorphism of $\FF$ and it is a linear map when $\FF$ is treated
as a vector space over the field $GF(2)$. Therefore, solving the equation $J(x)=y$
is tantamount to inverting a matrix over $GF(2)$. If $\FF=GF(2^m)$ then $\FF$
has dimension $m$, and $J$ is a $m\times m$ matrix of elements of $GF(2)$.
The only solution to equation~\eqref{eqn:quadratic-no-linear-term} is thus:
\begin{equation}
  \label{eqn:quadratic-no-linear-term-solution}
  x = J^{-1}\left(\frac{c}{a}\right).
\end{equation}
Of course, matrix $J$ may be precomputed, and matrix multiplication can be used to
find $x$. Alternatively, we can tabulate all square roots, and use a lookup table.

In the remainder of the paper we will write the unique solution to $x^2=y$ as $\sqrt{y}$.

We may assume that $a=1$ (i.e.\ make quadratic equation monic), by
dividing \eqref{eqn:quadratic-equation} by $a$.  As the next step, we
scale $x$, so that $b=1$. Let $x=b\,y$. We obtain:
\begin{align*}
  x^2+b\,x+c &= b^2\,y^2+b^2\,y+c= b^2\,y(y+1)+c,\\
  y(y+1) &= \frac{c}{b^2}.
\end{align*}
Hence, we reduced the quadratic equation~\eqref{eqn:quadratic-equation} to:
\begin{equation}
  \label{eqn:modified-quadratic-equation}
  y(y+1)=d
\end{equation}
where $d\in\FF$ is given and seek $y$. It turns out that this equation is also easy
to solve. If $d=0$ then $y=0$ or $y=1$. If $d\neq 0$ then we write the equation
in terms of the Frobenius automorphism $J$:
\begin{equation*}
  y(y+1)=y^2+y = J(y)+y = (J+I)\,y
\end{equation*}
where $I:\FF\to\FF$ is the identity map. We note that $J+I$ is a
linear transformation over $GF(2)$ given by a square matrix. We can
solve the equation $(J+I)\,y=d$ as a linear system of $m$ equations in
$m$ unknowns over $GF(2)$. We note that
\begin{equation*}
  \ker(J+I)=\{y\,:\, (J+I)\,y=0\} = \{0,1\}\subseteq\FF,
\end{equation*}
as $(J+I)\,y=0$ is equivalent to $y(y+1)=0$. Hence, the kernel is
$1$-dimensional as a linear subspace of the vector space $\FF$ over
the field $GF(2)$. This implies that the solution set of $(J+I)\,y=d$ is
a $1$-dimensional coset $y_0+\ker(J+I)$, where $y_0$ is a particular
solution, or no solution exists.  As $1$-dimensional subspaces in
characteristic $2$ are $2$-element sets, there are thus exactly $2$
solutions to the equation $y(y+1)=d$ for every $y$. Moreover, if the
two roots are $y_1$ and $y_2$ then $y_1-y_2\in \{0,1\}$
(i.e.\ $y_1-y_2\in\ker(J+I)$). Hence, if $y_1(y_1+1)=d$ then
$y_2=y_1+1$ is the second solution, and thus $y_1y_2=d$.

From the above discussion it follows that $J+I$ has nullity $1$ and
thus by the Rank-Nullity~Theorem its rank is $m-1$, as the dimension
of $\FF$ as a vector space over $GF(2)$ is $m$. Thus $\image(J+I)$ has
codimension $1$. Hence, the equation $y(y+1)=d$ has a solution for
exactly a half of the elements $d\in\FF$. Linear algebra tells 
us that there exists a linear functional $\varphi:\FF\to GF(2)$ such that
\begin{equation*}
  y(y+1)=d\quad \text{has a solution, iff}\quad \varphi(d)=0.
\end{equation*}
Moreover, $\varphi$ is the unique non-zero solution to
$(J^*+I)\varphi=0$, where $J^*:\FF^*\to \FF^*$ is the dual operator of
$J$, acting on the dual space $\FF^*$.  Equivalently, $\varphi$ is a
linear functional such that $\varphi\circ(J+I)=0$.

The above somewhat abstract discussion can be made concrete, if
$\FF=GF(2^m)$, and $D$ is a primitive element satisfying the equation
$g(D)=0$, where $g$ is the chosen primitive polynomial
$g(x)=\sum_{j=0}^{m}\,g_jx^j$ ($g_j\in GF(2)$ for $j=0,1,\ldots,m$ and
$g_m=1$). The set $\{1,D,D^2,\ldots,D^{m-1}\}$ is a basis of $\FF$ as
a vector space over $GF(2)$. This basis is used to identify
$d=\sum_{j=0}^{m-1} d_jD^j$ with a vector $(d_0,d_1,\ldots,d_{m-1})$
in the vector space $GF(2)^m$.  The linear map $J$ is represented with
respect to this basis by the matrix $\left(J_{ij}\right)$,
$i,j=0,1,\ldots,m-1$, where the entries $J_{ij}$ are found from the
formula:
\begin{equation}
  D^{2j}=\sum_{i=0}^{m-1}J_{ij}\,D^{i},\quad i,\;j=0,1,\ldots,m-1.
\end{equation}
Let $z$ be the left eigenvector of $J$ for eigenvalue $1$,
i.e.\ $(J+I)^\tr\,z=0$. In characteristic $2$ this vector is
unique. The functional $\varphi$ is then given by
$\varphi(x)=z^\tr\,x$. The condition of solvability of $(J+I)\,y=d$ is
$\varphi(z)=z^\tr d=0$. If $d=\sum_{j=0}^{m-1} d_jD^j$, where $D$ is
the primitive element generating $\FF$, and
$z=(z_0,z_1,\ldots,z_{m-1})$ then the condition of solvability is
\begin{equation}
  \label{eqn:quadratic-solvability-condition-1}
  \sum_{j=0}^{m-1}z_j\,d_j=0.
\end{equation}
Let $Z=\{j\in\{0,1,\ldots,m\}\,:\, z_j=1\}$ be the set of locations of
all non-zero coordinates of $z$ (counting from $0$). Then the
condition~\ref{eqn:quadratic-solvability-condition-1} can also be
written as:
\begin{equation}
  \label{eqn:quadratic-solvability-condition-2}
  \sum_{j\in Z}d_j=0.
\end{equation}
Thus, we may view the condition of solvability as a kind of parity
check on a subset of the coefficients $d_j\in\{0,1\}$ representing
$d\in\FF$.

\begin{example}[Operator $J+I$ for $m=4$]
  If $\FF=GF(2^4)$, $g(x)=x^4+x+1$, we find that the matrix of $J+I$
  is:
  \begin{equation*}
    J=
    \begin{pmatrix}
      1 & 0 & 1 & 0\\
      0 & 0 & 1 & 0\\
      0 & 1 & 0 & 1\\
      0 & 0 & 0 & 1
    \end{pmatrix},\qquad
    J+I=
    \begin{pmatrix}
      0 & 0 & 1 & 0\\
      0 & 1 & 1 & 0\\
      0 & 1 & 1 & 1\\
      0 & 0 & 0 & 0
    \end{pmatrix}.
  \end{equation*}
  If $D$ is the primitive element then the columns of $J$ are the
  coefficients of polynomials obtained by formally dividing
  $1,D^2,D^4,D^6$ by $g(D)$ and writing the coefficients of the
  remainder in ascending order of powers. Since $1,D^2$ are powers of
  $D$ below the degree of the primitive polynomial, they yield first
  two columns $(1,0,0,0)$ and $(0,0,1,0)$. Furthermore, long division
  over $GF(2)$ yields:
  \begin{align*}
    D^4 &= D+1 \mod g(D)\\
    D^6 &= D^2\cdot D^4=D^2\cdot (D+1)=D^3+D^2 \mod g(D).
  \end{align*}
  Hence, the $3^{rd}$ column of $J$ is $(1,1,0,0)$ and the $4^{th}$ is
  $(0,0,1,1)$.  We find $z$ by solving $(J+I)^\tr z=0$. We observe
  that $J+I$ has a row of zeros, so $(J+I)^\tr$ has a column of zeros, the
  $4^{th}$ column. Hence $z=(0,0,0,1)$ is a solution, and, on general
  grounds, this solution is unique. Hence, the condition of
  solvability of $y(y+1)=d$, where $d=d_0+d_1\,D+d_2\,D^2+d_3\,D^3$,
  is: $d_3=0$.
\end{example}
\begin{example}[Solvability for $m=8$ and other values of $m$]
  The most commonly used field is $\FF=GF(2^8)$ with
  $g(x)=x^8+x^4+x^3+x^2+1$. The condition of solvability of
  $y(y+1)=d$, where $d=\sum_{j=0}^7d_j\,D^j$, is $d_5=0$ and is found
  by the same approach. For some other values of $m$, we obtain: $m=5$
  with $g(x)=x^5+x^2+1$ yields $d_1+d_3=0$, $m=6$ with $g(x)=x^6+x+1$
  yields $d_5=0$. The case of $m=5$ demonstrates that the parity check
  $z^\tr d=0$ may involve more than $1$ coefficient $d_j$.
\end{example}

\section{The Decoding Algorithm}
\label{sec:decoding-algorithm}
In this section we will describe in detail the decoding algorithm for
the code identified by the generator matrix $G$ and parity check matrix
$H$, given by equations \eqref{eqn:generator-matrix},
\eqref{eqn:parity-check-matrix} and \eqref{eqn:parity-matrix}.

The decoder algorithm depends on the number of zeros and the patterns
we have noticed on the elements of the syndromes. We keep it in mind that
$\alpha_i$ are non-zero, distinct elements of the Galois field $\FF$.

Given the parity check matrix $H$ and the transmitted vector $t$,
error vector $e$ and received message vector $r=t+e$, we consider
the syndrome vector:
\begin{equation*}
  s=H\cdot r = H\cdot(t+e)=H\cdot e =
  \begin{pmatrix}
    s_1 &
    s_2 &
    s_3 &
    s_4 &
    s_5
  \end{pmatrix}^\tr
\end{equation*}

\subsection{The case of the zero syndrome vector}
Let us dispose of the easiest case of decoding first, that of $s=0$. 
\begin{lemma}[On the zero syndrome vector]
  If all the entries of the syndrome are zeros, then either we have no
  errors or we have silent data corruption that is not detected.
  Silent data corruption is possible when the number of errors is at
  least $5$.
\end{lemma}
\begin{proof}
  Since we have $H\cdot r=H\cdot e=0$, $r$ is a valid codeword
  (belongs to the columnspace of $G$).  If $r\neq t$ then we must have
  $H\cdot r = H\cdot t$, i.e.\ $H\cdot (r-t)=0$. Therefore, $r-t$ is a
  vector of weight at least $5$.  i.e.\ at least $5$ errors occurred.
\end{proof}

\subsection{A brief survey of syndrome decoding}
Syndrome decoding depends on the following observation: if
$i_1,i_2,\ldots,i_q$ are locations of the failed drives then we may be
able to determine the non-zero error values $e_{i_j}$ , and study a subsystem
of $H\cdot e = s$
\begin{equation}
\label{eqn:location-restricted-system}
\left(\begin{array}{c|c|c|c}
  H_{i_1} & H_{i_2} & \ldots & H_{i_q}
\end{array}\right)
\cdot 
\left(\begin{array}{c}
  e_{i_1} \\ e_{i_2} \\ \vdots \\ e_{i_q}
\end{array}\right)
= s
\end{equation}
where $H_{i}$ denotes the $i$-th column of $H$. We also use the
consequence of the fact that the code is systematic, that
\begin{equation*}
  H_{i}=\begin{cases}
  P_{i} & \text{$1\leq i\leq k$},\\
  I_{j} & \text{$i=k+j$, $j=1,2,\ldots,5$}.
  \end{cases}
\end{equation*}
where $P_{i}$ denotes the $i$-th column of the parity matrix $P_{5\times k}$, and
$I_{j}$ denotes the $j$-th column of the identity matrix $I_{5\times 5}$.
We recall that
\begin{equation*}
  P_{i}=
  \left(\begin{array}{c}
    1 \\ \alpha_i \\ \alpha_i^2 \\ \alpha_i^3 \\ \alpha_i(\alpha_i+1)
  \end{array}\right),
  \quad
  I_1=\left(\begin{array}{c}
    1 \\ 0 \\ 0 \\ 0 \\ 0
  \end{array}\right),
  \quad
  I_2=\left(\begin{array}{c}
    0 \\ 1 \\ 0 \\ 0 \\ 0
  \end{array}\right),
  \quad
  \ldots,
  \quad
  I_5=\left(\begin{array}{c}
    0 \\ 0 \\ 0 \\ 0 \\ 1
  \end{array}\right).
\end{equation*}
If the number of failed data drives is $r\leq q$ then we can split the
set of locations into $\{i_1,i_2\ldots,i_r\}$ and,
$\{k+j_1,k+j_2,\ldots,k+j_{q-r}\}$, where $j_l=i_{k+l}-k$ for
$l=1,2,\ldots,q-r$. Therefore, the linear
system~\eqref{eqn:location-restricted-system}
can be further specialized for systematic codes as:
\begin{equation}
  \label{eqn:specialized-linear-system-full-form}
\left(\begin{array}{c|c|c|c}
  P_{i_1} & P_{i_2} & \ldots & P_{i_r}
\end{array}\right)
\cdot 
\left(\begin{array}{c}
  e_{i_1} \\ e_{i_2} \\ \vdots \\ e_{i_r}
\end{array}\right)
+
\left(\begin{array}{ccc}
  e_{j_1} \\ e_{j_2} \\ \vdots \\ e_{j_{q-r}}
\end{array}\right)
= s.
\end{equation}
We will abbreviate this system to 
\begin{equation}
  \label{eqn:specialized-linear-system}
  \widetilde{P}\cdot e_{data} + e_{parity} = s
\end{equation}
where $\widetilde{P}$ is the submatrix obtained from $P_{5\times k}$ by keeping
only columns at failed drive locations, $i_1,i_2,\ldots,i_r$.

\subsection{The case of two failed parity drives}
The next observation is that, since we are interested only in
reconstructing of only up to 2 drives, if there is a failed parity
drive $j$ then the vector $\tilde{s}=s-e_{k+j}I_j$  is either $0$
or is proportional to one of the vectors $P_i$ or $I_l$ ($l\neq j$).
In the latter case,
\begin{equation}
  0=\tilde{s}-e_{k+l}I_l =s- e_{k+j}I_j-e_{k+l}I_l.
\end{equation}
i.e.\ $s=e_{k+j}I_j+e_{k+l}I_l$. This means that the failed parity drives
are $j$ and $l$, and their error values $e_{k+j}$, $e_{k+l}$ are found
from these simple equations:
\begin{align*}
  e_{k+j} &= s_{k+j},\\
  e_{k+l} &= s_{k+l}.\\
\end{align*}
Hence, we disposed of the case when both failed drives are parity
drives.

\subsection{Failure of one parity and one data drives}
Let us suppose that the failed drives are the $i$-th data
drive $i\leq k$ and $(k+j)$-th (parity) drive.  We will simply say that the $j$-th parity drive failed.
We have the following equation relating errors and syndromes
\begin{equation}
e_i\,P_{i}+e_{k+j}I_j=s
\end{equation}
where $I_j$ is the $j$-th column of the identity matrix
$I_{5\times5}$, and $e_i\neq 0$.  Thus, $e_i\alpha_i^{l-1}=s_l$ for
$l\neq j$, $l=1,2,3,4$, and $e_5\alpha_i(\alpha_i+1)=s_5$ if $5\neq
j$.  In particular $s_l\neq 0$ for $l=1,2,3,4$, $l\neq j$. Thus, if
$s_l=0$ for some $l\in\{1,2,3,4\}$ then automatically $j=l$.

It is also true that one of the following holds:
\begin{enumerate}
\item $\alpha_i\neq 1$ and $s_5=e_i\alpha_i(\alpha_i+1)$;
\item $\alpha_i=1$ and $s_5=0$.
\end{enumerate}
Using this information, we may find $\alpha_i$ as follows:
\begin{enumerate}
\item If $j=1$ then $\alpha_i=s_3/s_2$, $s_4/s_3=\alpha_i=s_3/s_2$ and $s_5/\alpha_i=s_5/(s_3/s_2)=s_5\,s_2/s_3=\alpha_i+1=s_3/s_2+1$.
\item If $j=2$ then $\alpha_i=s_4/s_3$, $s_3/s_1=\alpha_i^2$, $s_5/\alpha_i=s_5/(s_4/s_3)=s_5\,s_3/s_4=\alpha_i+1=s_4/s_3+1$; 
\item If $j=3,4,5$ then $\alpha_i=s_2/s_1$; $s_l/s_1 = \alpha_i^{l-1} = (s_2/s_1)^{l-1}$ if $l\neq j$, $l<5$; also $s_5/s_2=\alpha_i+1=(s_2/s_1)+1$ if $j\neq 5$;
\end{enumerate}
Once we have found $\alpha_i$, we set $e_l=s_l/\alpha_i^{l-1}$ using
one of the $l$ values found.  Then we set
$e_{k+j}=s_j-\alpha_i^{j-1}e_i$ if $j\leq 4$, or
$e_{k+5}=s_5-\alpha_i(\alpha_i+1)\,e_i$ if $j=5$.

Algorithm~\ref{alg:one-parity-one-data-location} implements the above
method for finding the location of the failed parity and data drives,
based on the syndrome vector $s$. It solves a slightly more general
equation:
\begin{equation}
  \label{eqn:one-parity-one-data-location}
  x\,\P{\rho}+y\,I_j = s.
\end{equation}
It not only finds $x$ and $y$, but also $j$ and $\rho$, which is
crucial to finding the locations of the failed drives.  The algorithm
returns the quadruple $(j,\rho,x,y)$.
Algorithm~\ref{alg:one-parity-one-data-location} rejects solutions in
which $x=0$ or $y=0$.  Thus, it does not handle the (easy) case
when $s$ has weight $1$, which should be treated separately. Hence,
the input vector $s\in\FF^5$ should have weight at least $2$. It
should be noted that for some syndrome vectors $s$
\emph{two~solutions} of
equation~\eqref{eqn:one-parity-one-data-location} exist and have
$\rho\neq 0$. This can only occur when $\rho$ is a cubic root of unity
$\neq 1$. Since we excluded one of the cubic roots from the set
$\alpha$, the $\rho$ which belongs to $\alpha$ is
unique. Algorithm~\ref{alg:one-parity-one-data-location} needs the
explicit knowledge of the excluded root, so that it can omit it in its
search for solutions. We pass the excluded cubic root of unity in the
second argument to the function
\textproc{LocateFailedParityAndData}{$(s,X)$}, where $X$ should be either
empty set or a 1-element set containing the excluded root.

Algorithm~\ref{alg:one-parity-one-data-recovery} uses $j$, $\rho$, $x$ and $y$
found by Algorithm~\ref{alg:one-parity-one-data-location} to compute
the error vector, as indicated above. 

\begin{algorithm}[htb]
  \caption{\label{alg:one-parity-one-data-location} This algorithm
    yields the solution $(j,\rho,x,y)$ of the equation
    $x\,\P{\rho}+y\, I_j=s$, where $\P{\rho}$ is a vector given by
    equation~\eqref{eqn:P-rho}, $I_j$ is the $j$-th column of
    $I_{5\times 5}$ and $x,y\in\FF$.  The first input is a vector
    $s\in \FF^{5}$.  The second input is a set (possibly empty) $X$,
    consisting of cubic roots of unity which are not acceptable values
    of $\rho$. If there is a (unique) solution such that $x\neq 0$ and
    $\rho\neq 0$, the quadruple $(j,\rho,x,y)$ is returned.  If there
    is no solution with these properties, $j=0$ and $(0,0,0,0)$ is
    returned. In particular, if there is a solution such that $y=0$,
    i.e.\ $x\,\P{\rho}=s$, then $j=6$ and the quadruple $(6,\rho, x,
    0)$ is returned.  If $s$ is of weight $1$ then there is a solution
    to $y\,I_j=s$, but this case is also treated as failure, $j$ is
    set to $0$ and $(0,0,0,0)$ is returned. Any $\rho$ which may be
    otherwise be a solution, but belongs to $X$, will not be
    considered a solution, and another $\rho$ will be tried.  }
  \begin{algorithmic}[1]
    \Function{LocateFailedParityAndData}{$s$,$X$}
    \State $j\gets 0$
    \State $\rho\gets 0$\Comment{Zero in $\FF$.}
    \State $x\gets 0$\Comment{Zero in $\FF$.}
    \State $y\gets 0$\Comment{Zero in $\FF$.}
    \For {$m=1,3$}\Comment{If a solution exists, $s_1\cdot s_2\neq 0$ or $s_3\cdot s_4\neq 0$.}
    \If{$s_m=0\Or s_{m+1}=0$}
    \State\Continue\Comment{Hence, if $m=1$, try $m=3$.}
    \EndIf
    \State $\rho\gets s_{m+1}/s_m$\Comment{A candidate for $\rho$, $\rho\neq 0$.}
    \State $P\gets\left(1,\rho,\rho^2,\rho^3,\rho(\rho+1)\right)$\Comment{Note: $P=\P{\rho}$.}
    \If{$P_4=1 \And \rho\in X$}\Comment{Potential $\rho$ is an excluded root of unity.}
    \State \Continue;\Comment{Try another $\rho$ which may exist and is not excluded.}
    \EndIf
    \State $x\gets s_m/P_m$\Comment{A candidate for $x$, $x\neq 0$.}
    \State $failcount\gets 0$\Comment{Initialize failed parity drive count.}
    \For {$t=1,2,3,4,5$, such that $t\neq m, m+1$}
    \If {$s_t\neq x\,P_t$}\Comment{Consistency violation with $s=x\,P$ at position $t$.}
    \State $failcount\gets failcount+1$\Comment{Increment failed parity drive count.}
    \State $j\gets t$\Comment{Make parity drive $t$ a candidate for a failed parity drive.}
    \EndIf
    \EndFor
    \If {$failcount=0$} 
    \State $j\gets 6$ \Comment{No failed parity drive.}
    \State\Goto{one-parity-one-data-location:end}
    \ElsIf{$failcount>1$} \Comment{More than $1$ failure, no solution.}
    \State $j\gets0$\Comment{Signal failure, or try another $\rho$.}
    \Else\Comment{$failcount=1$.}
    \State $y\gets s_j-x\cdot P_j$\Comment{Parity drive $j$ is the only inconsistent parity drive.}
    \State\Goto{one-parity-one-data-location:end}
    \EndIf
    \EndFor
    \State\Return{$(j,\rho,x,y)$} \label{one-parity-one-data-location:end}
    \EndFunction
  \end{algorithmic}
\end{algorithm}

\begin{algorithm}[htb]
  \caption{Recovering from a failure of one parity and one data drive
    at unknown locations, or a single data drive failure at unknown
    location.  The algorithm accepts as input the set $\alpha$ of
    elements of $\FF$ and a syndrome vector $s\in\FF^5$. It finds the
    locations $i$ and $j$ of the failed drives, and determines the
    error vector $e\in\FF^{k+5}$ such that $e_i\,P_i+e_{k+j}I_j=s$.
    Upon success, it returns $(e,\True)$. If no solution exists,
    $(0,\False)$ is returned.  The algorithm handles correctly the
    case of a single data drive failure, by finding a solution to
    $e_i\,P_i=s$ and setting $e_{k+j}=0$ for $j=1,2,3,4,5$.
    \label{alg:one-parity-one-data-recovery}}
  \begin{algorithmic}[1]
    \Function{RecoverFailedParityAndData}{$\alpha,s$}
    \State $k\gets\Call{NumberOfElements}{\alpha}$
    \State $e\gets 0$           \Comment{$0$ in $\FF^{k+5}$.}
    \State $X\gets\Call{ExcludedCubicRootsOfUnity}{\alpha}$
    \State $(j,\rho,x,y)\gets\Call{LocateFailedParityAndData}{s,X}$
    \If{$j=0$}\Comment{$s$ does not come from a single parity, single data drive failure.}
    \State \Return{$(e,\False)$}\Comment{Still, $e=0$.}
    \EndIf
    \State $i\gets\Call{Lookup}{\alpha,\rho}$\Comment{Find $i$ such that $\alpha_i=\rho$}
    \If {$i=\emptyset$}
    \State \Return{$(e,\False)$}\Comment{Still, $e=0$.}
    \EndIf
    \State $e_i\gets x$
    \If {$j\leq 5$}
    \State $e_{k+j}\gets y$
    \EndIf
    \State \Return $(e,\True)$
    \EndFunction
  \end{algorithmic}
\end{algorithm}

\subsection{Failure of two data drives}
The only cases left to consider are those in which the only failed
drives are data drives. Not surprisingly, this is the most delicate
case to analyze.

We will somewhat relax the above assumption, by assuming only that up
to $2$ data drives have failed. If the locations of the failed drives
(if any) are at locations $i$ and $j$, the linear
system~\eqref{eqn:specialized-linear-system} is:
\begin{equation}
  \label{eqn:two-drive-linear-system}
  \left(\begin{array}{cc}
    1                   & 1                   \\
    \alpha_i            & \alpha_j            \\
    \alpha_i^2          & \alpha_j^2          \\
    \alpha_i^3          & \alpha_j^3          \\ 
    \alpha_i(\alpha_i+1)& \alpha_j(\alpha_j+1)
  \end{array}\right)
  \cdot
  \left(\begin{array}{c}
    e_{i}\\
    e_{j}\\
  \end{array}\right)
  = 
  \begin{pmatrix}
    s_1 \\
    s_2 \\
    s_3 \\
    s_4 \\
    s_5
  \end{pmatrix}
\end{equation}
The problem reduces to the following: given $s$, find the locations
$i$ and $j$, and the error values $e_i$ and $e_j$.

The most helpful result comes from linear algebra:
\begin{theorem}[A criterion of solvability of a linear system]
  A linear system $A\cdot x=b$ has a solution iff $N^Tb=0$, where $N$
  is a matrix whose columns form a basis of the nullspace of $A^\tr$.
\end{theorem}
This theorem is typically stated as $\image(A)^\perp = \ker(A^\tr)$.

We proceed to calculate the basis of the nullspace of $A^\tr$. Using
elementary matrices, we find the reduced row echelon form without row
exchanges:
\begin{align*}
  &\begin{pmatrix}
    1 & \alpha_i \\
    0 & 1
  \end{pmatrix}
  \begin{pmatrix}
    1 & 0 \\
    0 & \alpha_j+\alpha_i
  \end{pmatrix}^{-1}
  \begin{pmatrix}
    1 & 0 \\
    1 & 1
  \end{pmatrix}
  \begin{pmatrix}
    1 & \alpha_i & \alpha_i^2 & \alpha_i^3 & \alpha_i(\alpha_i+1)\\
    1 & \alpha_j & \alpha_j^2 & \alpha_j^3 & \alpha_j(\alpha_j+1)
  \end{pmatrix}\\
  &=
  \begin{pmatrix}
    1 & 0 &\alpha_i\,\alpha_j &\alpha_i\,\alpha_j(\alpha_i+\alpha_j) &\alpha_i\,\alpha_j\\
    0 & 1 & \alpha_j+\alpha_i  & \alpha_i^2+\alpha_i\,\alpha_j+\alpha_j^2& \alpha_i+\alpha_j+1
  \end{pmatrix}
\end{align*}
A simple rearrangement of the entries of this matrix yields:
\begin{equation*}
  N=\begin{pmatrix}
    \alpha_i\,\alpha_j    &   \alpha_i\,\alpha_j(\alpha_i+\alpha_j)  &  \alpha_i\,\alpha_j \\
    \alpha_i+\alpha_j     &  \alpha_i^2+\alpha_i\,\alpha_j+\alpha_j^2    & \alpha_i+\alpha_j+1 \\
    1                    &   0                                          &  0                  \\
    0                    &   1                                          &  0                  \\ 
    0                    &   0                                          &  1          
  \end{pmatrix}.
\end{equation*}
Thus $\image(N)=\ker(A^\tr)$. Now the sufficient and necessary condition of solvability of $A\cdot x=b$ is $N^\tr b=0$,
or, after some simple transformations:
\begin{equation*}
  \left\{
  \begin{array}{ccl}
    s_{3} &=& s_1\,\alpha_i\,\alpha_j + s_2(\alpha_i+\alpha_j)\\
    s_{4} &=& s_1\alpha_i\,\alpha_{j}(\alpha_i+\alpha_j)+s_2(\alpha_i^2+\alpha_i\,\alpha_j+\alpha_j^2) \\
    s_{5}-s_2 &=& s_1\,\alpha_i\,\alpha_j+s_2(\alpha_i+\alpha_j)
  \end{array}\right.
\end{equation*}
Turning to the first and last equation in this system, we observe that
they have the same right-hand side. Thus, a necessary condition of
solvability is $s_5-s_2=s_3$ or $s_2+s_5+s_3=0$. We note that this is
the equation which involves the parity check involving only the parity
drives: the $5^{th}$ row of the parity matrix is the sum of the $2^{nd}$
and $3^{rd}$ rows.

If $s_2+s_5+s_3=0$ then the last equation can be discarded, and we obtain the following system:
\begin{equation*}
  \left\{
  \begin{array}{ccl}
    s_{3} &=& s_1\,\alpha_i\,\alpha_j+s_2(\alpha_i+\alpha_j)\\
    s_{4} &=& s_1\alpha_i\,\alpha_{j}(\alpha_i+\alpha_j)+s_2(\alpha_i^2+\alpha_i\,\alpha_j+\alpha_j^2)
  \end{array}\right.
\end{equation*}
from which we are able to find $\alpha_i$ and $\alpha_j$. We can express
this system in terms of the symmetric polynomials:
\begin{equation}
  \label{eqn:symmetric-polynomial-substitution}
  \left\{
  \begin{array}{ccc}
    \sigma_1 &=& \alpha_i +\alpha_j,\\
    \sigma_2 &=& \alpha_i\cdot\alpha_j.
  \end{array}\right.
\end{equation}
It should be noted that we are only interested in solution where $\alpha_i\neq\alpha_j$, so $\sigma_1\neq 0$.
Also, $\sigma_2=\alpha_i\alpha_j\neq0$ as all $\alpha_i$ are non-zero. In short,
we are also interested in vectors $(\sigma_1,\sigma_2)\neq 0$.

We utilize the Frobenius identity to rewrite $\alpha_i^2+\alpha_j^2$ in the second equation as
$\sigma_1^2$. We obtain
\begin{equation*}
  \left\{
  \begin{array}{ccl}
    s_{3} &=& s_1\,\sigma_2+s_2\,\sigma_1 \\
    s_{4} &=& s_1\,\sigma_1\,\sigma_2+s_2\,(\sigma_1^2+\sigma_2)
  \end{array}\right.
\end{equation*}
The second equation can also be written as
$s_{4}=(s_1\,\sigma_2+s_2\,\sigma_1)\,\sigma_1 +
s_2\,\sigma_2=s_3\,\sigma_1 + s_2\,\sigma_2$. Using the first
equation, we obtain the system:
\begin{equation*}
  \left\{
  \begin{array}{ccl}
    s_{3} &=& s_1\,\sigma_2 + s_2\,\sigma_1 \\
    s_{4} &=& s_3\,\sigma_1 + s_2\,\sigma_2 
  \end{array}\right.
\end{equation*}
This is a linear system for $(\sigma_1,\sigma_2)$ which in matrix form is:
\begin{equation}
  \label{eqn:symmetrized-system-2}
  \begin{pmatrix}
    s_2 & s_1 \\
    s_3 & s_2
  \end{pmatrix}
  \begin{pmatrix}
    \sigma_1\\
    \sigma_2
  \end{pmatrix}
  =
  \begin{pmatrix}
    s_3 \\
    s_4
  \end{pmatrix}
\end{equation}
The determinant of the matrix of the system is $D=s_2^2-s_1\,s_3$. Therefore, we
have two distinct cases: $D=0$ and $D\neq 0$. Let us analyze the singular
case $D=0$ first.
\begin{lemma}[On the singular case $D=0$]
  Let us assume that the system~\eqref{eqn:two-drive-linear-system} is
  consistent (has a solution).  Let $D=s_2^2-s_1\,s_3$ be the
  determinant of the matrix of the system
  ~\eqref{eqn:symmetrized-system-2}.  If $D=0$, where then either $s=0$
  or $s$ comes from a single failed data drive $i$ and $s=s_1\,P_i$.
\end{lemma}
\begin{proof}
  The necessary condition for a solution
  of~\eqref{eqn:symmetrized-system-2} to exist is that the two
  determinants $D_1$ and $D_2$ are also zero, where $D_1$ and $D_2$ come
  from Cramer's rule vanish:
  \begin{equation*}
    D_1=\left| \begin{matrix} s_3 & s_1 \\ s_4 & s_2 \end{matrix}\right|=s_2\,s_3-s_1\,s_4=0,
    \qquad
    D_2=\left| \begin{matrix} s_2 & s_3 \\ s_3 & s_4 \end{matrix}\right|=s_2\,s_4-s_3^2=0.
  \end{equation*}
  Thus, we have a system of equations:
  \begin{equation}
    \label{eqn:degeneracy-condition}
    \left\{
    \begin{array}{ccc}
      s_1\,s_3  &=& s_2^2  \\
      s_2\,s_3  &=& s_1\,s_4 \\
      s_2\,s_4  &=& s_3^2
    \end{array}\right.
  \end{equation}
  Let us analyze possible solutions of~\eqref{eqn:degeneracy-condition}. 

  If $s_1=0$ then $s_2=s_3=0$. Also, $s_5=s_2+s_3=0$. Hence, $s_4$ can
  be the only non-zero syndrome.  The first two equations of the
  system~\eqref{eqn:two-drive-linear-system}  
  \begin{equation*}
    \left(\begin{array}{cc}
      1                   & 1                   \\
      \alpha_i            & \alpha_j            \\
    \end{array}\right)
    \cdot
    \left(\begin{array}{c}
      e_{i}\\
      e_{j}\\
    \end{array}\right)
    = 
    \begin{pmatrix}
      0 \\ 0
    \end{pmatrix}
  \end{equation*}
  with determinant
  $\alpha_j-\alpha_i\neq 0$ imply that $e_i=e_j=0$ (note: $\alpha_i\neq\alpha_j$). Hence,
  $s_4=0$ and $s=0$.

  Let us thus assume $s_1\neq 0$. If $s_2=0$ then $s_3=s_4=s_5=0$.
  Hence $s_1$ would be the only non-zero syndrome. The second and third
  equation of the system~\eqref{eqn:two-drive-linear-system} 
  \begin{equation*}
    \left(\begin{array}{cc}
      \alpha_i            & \alpha_j            \\
      \alpha_i^2          & \alpha_j^2          \\
    \end{array}\right)
    \cdot
    \left(\begin{array}{c}
      e_{i}\\
      e_{j}\\
    \end{array}\right)
    = 
    \begin{pmatrix}
      0 \\ 0
    \end{pmatrix}
  \end{equation*}
  with determinant
  $\alpha_i\alpha_j^2-\alpha_j\alpha_i^2=\alpha_i\alpha_j(\alpha_i-\alpha_j)\neq0$
  imply that $e_i=e_j=0$ (note: $\alpha_i\neq\alpha_j$).

  If $s_1\neq 0$ and $s_2\neq 0$ then $s_3\neq 0$ and $s_4\neq0$. Hence,
  $s_1$, $s_2$, $s_3$ and $s_4$ are all non-zero. We find from the first
  equation of~\eqref{eqn:degeneracy-condition} that
  $s_3=s_2^2/s_1$. Plugging into the second and third equations, we get
  $s_2(s_2^2/s_1)=s_1\,s_4$ and $s_2\,s_4=(s_2^2/s_1)^2$. Simplifying,
  $s_2^3=s_1^2\,s_4$ and $s_1^2\,s_4=s_2^3$. Both equations are identical.

  In summary, if the linear system~\eqref{eqn:symmetrized-system-2} is
  singular then either $s=0$ and $e_i=e_j=0$, or $s\neq 0$ and 
  we have both $s_3\,s_1=s_2^2$ and $s_1^2s_4=s_2^3$. In the latter case,
  $s_1$, $s_2$, $s_3$ and $s_4$ are all non-zero. We rewrite these equations
  as:
  \begin{equation*}
    \dfrac{s_2}{s_1}=\dfrac{s_3}{s_2}=\dfrac{s_4}{s_3}
  \end{equation*}
  (Note: $s_4/s_3=(s_2^3/s_1^2)/(s_2^2/s_1)=s_2/s_1$).  Let $\rho$ be
  the common value of these ratios. Then 
  \begin{equation*}
    s_5=s_2+s_3 =  s_1\left((s_2/s_1)+(s_3/s_2)\,(s_2/s_1)\right)=s_1\rho(\rho+1).
  \end{equation*}
  Hence, there is a $\rho\in\FF^*$ such that $s$ is proportional to
  the vector $\P{\rho}$ defined by equation~\eqref{eqn:P-rho}.
  If $\rho=\alpha_i$ for some $i$ then $s=s_1\,P_i$. This implies that
  $e$ has weight $1$ and $e_i=s_1$, and in fact only one data drive
  failed. If $\rho\neq\alpha_i$ for all $\alpha_i$ then $s=s_1\P{\rho}$
  cannot be a linear combination of $P_i$ and $P_j$ for any combination
  of $i$ and $j$. Indeed, the matrix
  \begin{equation*}
    \left(\begin{array}{c|c|c}P_i&P_j&\P{\rho}\end{array}\right)
  \end{equation*}
  contains a non-singular $3\times 3$ Vandermonde matrix, and its columns are thus
  linearly independent. 
\end{proof}

Thus, we may assume $D\neq 0$ and obtain the solution of the
non-singular system~\eqref{eqn:symmetrized-system-2} by Cramer's rule:

\begin{equation}
  \label{eqn:two-drive-solution}
  \left\{
  \begin{array}{ccccc}
    \sigma_1 &=& \dfrac{D_1}{D} &=& \dfrac{s_2\,s_3-s_1\,s_4}{s_2^2-s_1\,s_3},\\
    \sigma_2 &=& \dfrac{D_2}{D} &=& \dfrac{s_2\,s_4-s_3^2}{s_2^2-s_1\,s_3}.
  \end{array}\right.
\end{equation}
Once we have found $\sigma_1$ and $\sigma_2$, we find the roots
$\alpha_i$ and $\alpha_j$, in view of the Vieta identity
$(\zeta-\alpha_1)\,(\zeta-\alpha_2)=\zeta^2-\sigma_1\,\zeta+\sigma_2$,
from the quadratic equation
\begin{equation*}
  \zeta^2-\sigma_1\,\zeta+\sigma_2=0.
\end{equation*}
This yields $i$ and $j$, the locations of the failed drives. The error values
can be obtained from the first two equations of~\eqref{eqn:two-drive-linear-system}:
\begin{equation*}
  \left(\begin{array}{cc}
    1                   & 1                   \\
    \alpha_i            & \alpha_j            \\
  \end{array}\right)
  \cdot
  \left(\begin{array}{c}
    e_{i}\\
    e_{j}\\
  \end{array}\right)
  = 
  \begin{pmatrix}
    s_1 \\ s_2
  \end{pmatrix}
\end{equation*}
Explicitly given, they are:
\begin{equation}
  \label{eqn:two-drive-error-values}
  \left\{
  \begin{array}{ccc}
    e_i &=& \dfrac{s_1\alpha_j-s_2}{\alpha_j-\alpha_i},\\ 
    e_j &=& \dfrac{s_2-\alpha_is_1}{\alpha_j-\alpha_i}.
  \end{array}\right.
\end{equation}

\begin{algorithm}[htb]
  \caption{This algorithm solves the equation $x\,\P{u}+y\,\P{v}=s$,
    where $\P{u},\P{v}$ are given by equation~\eqref{eqn:P-rho}. The
    input is the syndrome vector $s\in\FF^5$. If a solution exists, we
    return the quadruple $(u,v,x,y)$.  The algorithm handles only
    cases requiring $x$ and $y$ to be non-zero, i.e.\  two failed data
    drives. It is assumed that $s_2+s_3+s_5=0$, as otherwise the
    syndrome vector would imply at least one parity drive failure.  We
    also assume $D, D_1, D_2 \neq 0$, as otherwise the
    system~\eqref{eqn:two-drive-linear-system} has been shown to be
    either inconsistent, or $s=0$, or $s$ comes from a single data
    drive failure. This algorithm utilizes
    formulas~\eqref{eqn:two-drive-solution} and~\eqref{eqn:two-drive-error-values}.
    \label{alg:two-data-drive-location}}
  \begin{algorithmic}[1]
    \Function{LocateTwoFailedDataDrives}{$s$}
    \State $u\gets 0$
    \State $v\gets 0$
    \If {$s_2+s_3+s_5\neq 0$} \Comment{This syndrome implies a failure of parity drive.}
    \State\Goto{alg:two-data-drive-location:end}\Comment{Signal failure by returning $u=v=0$.}

    \EndIf
    \State $D\gets s_2^2-s_1\,s_3$ 
    \State $D_1\gets s_2\,s_3-s_1\,s_4$
    \State $D_2\gets s_2\,s_4-s_3^2$
    \If {$D=0$ \Or $D_1=0$ \Or $D_2=0$}\Comment{Not handled by this algorithm.}
    \State\Goto{alg:two-data-drive-location:end}\Comment{Signal failure by returning $u=v=0$.}
    \EndIf
    \State $\sigma_1 \gets D_1/D$\Comment{$\sigma_1\neq 0$}
    \State $\sigma_2 \gets D_2/D$\Comment{$\sigma_2\neq 0$}
    \State $\{u,v\}\gets \Call{SolveQuadraticEquation}{1,-\sigma_1,\sigma_2}$\Comment{Solve $\zeta^2-\sigma_1\,\zeta+\sigma_2=0$.}
    \State $x\gets (v\cdot s_1-s_2)/(v-u)$
    \State $y\gets (s_2-u\cdot s_1)/(v-u)$
    \State \Return{$(u,v,x,y)$}\label{alg:two-data-drive-location:end}
    \EndFunction
  \end{algorithmic}
\end{algorithm}

\begin{algorithm}[htb]
  \caption{Recovery from a failure of two data drives at unknown
    locations.  The input consists of the set $\alpha\subseteq\FF$ and
    the syndrome vector $s\in\FF^5$. Upon success, the algorithm
    returns the error vector $e\in\FF^{k+5}$, where $k$ is the number
    of elements of $\alpha$. The algorithm solves the equation
    $e_iP_i+e_jP_j=s$, when there is a solution with $e_i,e_j\neq 0$,
    i.e.\ the syndrome comes from two failing data drives, but does
    not come from either $s=0$ or a single drive failure (parity or
    data).  If there is no solution satisfying these properties
    $(0,\False)$ is returned.
    \label{alg:two-data-drive-recovery}}

  \begin{algorithmic}[1]
    \Function{RecoverTwoFailedDataDrives}{$\alpha$,$s$}
    \State $(u,v,x,y)\gets\Call{LocateTwoFailedDataDrives}{s}$
    \State $k\gets\Call{NumberOfElements}{\alpha}$
    \State $i\gets\Call{Lookup}{\alpha,u}$
    \State $j\gets\Call{Lookup}{\alpha,v}$
    \State $e\gets 0$\Comment{This is $0\in\FF^{k+5}$.}
    \If {$i\neq \emptyset \And j\neq\emptyset$}
    \State $e_i\gets x$
    \State $e_j\gets y$
    \State \Return ($e$,$\True$)
    \Else
    \State \Return ($e$,$\False$)
    \EndIf
    \EndFunction
  \end{algorithmic}
\end{algorithm}

\section{The Main Algorithm: Recovery from up to $2$ Errors (unknown locations errors)}
\label{sec:two_unknown}
Algorithm~\ref{alg:main} defines the overall flow control structure,
but does little work on its own. It uses several other algorithms
for which we do not define any pseudo-code, as they are straightforward
once a particular implementation strategy is chosen. Here we list
them with their signature and requirements:
\begin{enumerate}
\item A function \textproc{Lookup}{($\alpha$, $\zeta$)}, which returns an
  index $i$, $1\leq i\leq k$, such that $\alpha_i=\zeta$, or $\emptyset$ if
  $\zeta\notin\alpha$. Here
  $\alpha=\{\alpha_1,\alpha_2,\ldots,\alpha_k\}$ is the set of
  elements of the underlying Galois field of characteristic $2$. For
  example, $\alpha$ could be implemented as a map
  $\alpha:\{1,2,\ldots,k\}\to\FF$.
\item A function \textproc{FindNonZeros}{($s$)}, which returns a list of indices
  $i$ of a Galois vector $s\in\FF^{5}$ such that $s_i\neq 0$.
\item A function \textproc{ParityCheckMatrix}{($\alpha$)}, which returns the
  $5\times \left(k+5\right)$ parity check matrix $H$ of our code,
  defined by equation~\eqref{eqn:parity-check-matrix}.
\item A function \textproc{NumerOfElements}{($collection$)}, which returns the
  number of elements of generic collections of objects, such as sets,
  lists and vectors.
\item A function \textproc{SolveQuadraticEquation}{($a$,$b$,$c$)}, which
  returns the two roots of the equation $f(x)=a\,x^2+b\,x+c=0$, where
  $a,b,c\in\FF$ and $x$ is a variable ranging over $\FF$. We outlined
  two algorithms in Section~\ref{sec:quadratic-equation} for doing
  this.
\item A function \textproc{ExcludedCubicRootsOfUnity}{($\alpha$)}, which
  returns the list of roots of the equation $\zeta^2+\zeta+1=0$ which
  are \emph{not} in the set $\alpha$.  This list is possibly empty,
  and has not more than $1$ element, if our exclusion rules are
  observed.
\end{enumerate}

\begin{algorithm}[htb]
  \caption{The decoding algorithm for code defined by the parity
    matrix~\eqref{eqn:parity-matrix}.  The inputs are: a subset
    $\alpha$ of non-zero elements of the Galois field $\FF$ and the
    received vector $r$. The output is either $(t,\True)$, there $t$
    is the transmitted vector, or $(r,\False)$ if corruption of more
    than $2$ drives is detected. If no more than $2$ drives had an
    error, $t$ is guaranteed to be correct.\label{alg:main}}
  \begin{algorithmic}[1]
    \Function{RaidDecode}{$\alpha$,$r$}
    \State $H\gets\Call{ParityCheckMatrix}{\alpha}$ \Comment{Obtain parity check matrix.}
    \State $s\gets H\cdot r$ \Comment{Compute the syndrome}
    \State $lst\gets\Call{FindNonZeros}{s}$ \Comment{Get indices of non-zero elements of $s$.}
    \State $nz\gets\Call{NumberOfElements}{lst}$ \Comment{Find number of non-zeros (weight of $s$).}
    \If {$nz=0$} \Comment{Do nothing, no error detected.}
    \ElsIf {$nz=1$} \Comment{One parity drive failed.}
    \State $i\gets lst_1$
    \State $e_{k+i}\gets s_i$
    \ElsIf {$nz=2$} \Comment{Two parity drives failed.}
    \State $i\gets lst_1$
    \State $j\gets lst_2$
    \State $e_{k+i}\gets s_{i}$
    \State $e_{k+j}\gets s_{j}$
    \Else
    \State $(e,status)\gets\Call{RecoverFailedParityAndData}{\alpha, s}$
    \If {$status=\True$}\Comment{One parity and one data drive failed, and recovered.}
    \State \Goto{alg:main:endmark}
    \EndIf
    \State $(e,status)\gets\Call{RecoverTwoFailedDataDrives}{\alpha, s}$
    \If {$status=\True$}\Comment{Two data drives failed, and recovered.}
    \State \Goto{alg:main:endmark}
    \EndIf
    \State \Return $(r,\False)$\Comment{Return received message and signal failure.}
    \EndIf
    \State $t\gets r+e$\label{alg:main:endmark}\Comment{Compute the transmitted vector $t$ by correcting errors in $r$.}
    \State \Return $(t,\True)$\Comment{Return transmitted vector and signal success.}
    \EndFunction
  \end{algorithmic}
\end{algorithm}

\begin{remark}
Generalization to any field $GF(2^{m})$:

The algebra rules used in the current paper are applicable to any
field $GF(2^{m})$.  For $\FF=GF(2^8)$ the maximum number of drives
supported is $254$.  If we need to build a RAID array with more
than $254$ drives, we can choose $\FF=GF(2^{m})$ with $m>8$ and also
excludes one of the $3^{rd}$ roots of unity if exist, yielding limit of up to
$2^m-1$ or $2^m-2$ possible drives.

For example, If $\FF=GF(2^{16})$, $2^{16}=65,536$. We have to exclude
the zero element and one of the $3^{rd}$ roots of unity and have a
limit of up to to $65,534$ drives.
\qed\\
\end{remark}

\begin{remark}
What if we do not want to exclude one of the two $3^{rd}$ roots of unity?

Let us assume that we are using all of the 255 drives for
$\FF=GF(2^8)$. Then, the algorithm still functions and has a very low
probability of not working correctly!!

This algorithm will fail only in the case of having two failed data drives, whose locations correspond to both of the $3^{rd}$
root of unity, with equal error values.

Therefore, the probability that our algorithm fails 
due to non-exclusion of a $3^{rd}$ root of unity is

\begin{equation*}
  \frac{1}{{2^{m}-1\choose 2}\cdot\left(2^{m}-1\right)}.
\end{equation*}
(This is a conditional probability, under the assumption that failure
indeed occurs.)

For example, if $\FF=GF(256)$, the risk is $\approx10^{-7}$, and for the
$\FF=GF(2^{16})$, the risk would be $\approx 7.1 \times 10^{-15}$.
\qed\\
\end{remark}

\section{Computational Complexity}
\label{sec:complexity}
It is clear that algorithm~\ref{alg:main} \textbf{involves a constant (very
  small) number of Galois field operations} (additions,
multiplications and divisions) over the field $\FF$. If we choose to
solve quadratic equations using Gaussian elimination, the number of
operations in $GF(2)$ is $O((\log|\FF|)^3)$, which is the
computational complexity of Gaussian elimination, while the lookup
table approach is constant time. Thus, the algorithm corrects a single
stripe containing an error in constant time, independent of of the
size of the field, assuming lookup table implementation.

A more-in-depth analysis of computational complexity requires taking
into account the fact that with a growing number of disks we must
also allow the field to grow. If $\FF=GF(2^m)$, and $N = k + 5$ is the
total number of disks in the array, we must use a field for which
the number $Q=2^m$ is equal to $N$ up to $1$ or $2$ disks.

The complexity of a RAID method implementing striping typically is
computed as the time or space required to encode/decode a single
codeword, which is a stripe.  
As an example, RAID~6 requires a fixed number of Galois
field operations (addition, multiplication, division, logarithm
lookup) when decoding a received vector, not counting the computation
of parities or syndromes. This correlates with the number of CPU
cycles and the time required to decode a received vector. The number
of operations does not depend on the number of disks in the array $N$.
Also, constant time access is assumed to array elements. However, as
$N$ increases, it is necessary to use a larger Galois field $GF(Q)$
with $Q>N$. The number of bits $B=\lceil \log_2(Q)\rceil $ per Galois
field element grows, thus requiring more time per Galois field
operation. Addition, which is identical to XOR, is
$O(B)$. Multiplication has complexity
$O(B\cdot\log(B)\cdot \log(\log(B)))$, according to the state of the
art \cite{fast-multiplication}. Being close to $O(B)$, this kind of
complexity is referred to as quasilinear time complexity.  Thus, the
time complexity of RAID~6 decoding is
$O(\log(N)\cdot \log\log(N)\cdot\log\log\log(N))$ rather than
constant, and will be called quasi-logarithmic in the paper.  Some
operations, such as solving a quadratic equation in $GF(Q)$, where
$Q=2^m$, require inverting a matrix with coefficients in $GF(2)$ of
size $O(m)$. The complexity of matrix inversion is $O(m^p)$ where the
best $p\leq 3$, and the best known value of $p$ known today is
$p\approx 2.373$. Thus, the complexity in terms of $Q$ of solving
quadratic equation is $O(\log(Q)^p)$.  It is possible to choose $Q$
arbitrarily large, independently of the number of disks $N$, incurring
large computational cost.  With an optimal choice of $Q$,
$Q\leq 2\times N$ and the computational cost is
$O(\log(N)^p)$. However, the matrix inversion for solving quadratic
equation can be performed only once, with its result stored in a
lookup table of size $O(m^2)$, thus not affecting run time, assuming
lookup time $O(1)$ or even $O(\log m)$, if binary search needs to be
used.  Thus, having an algorithm which performs a constant number of
Galois field operations, and solves a fixed number of quadratic
equations, remains quasi-logarithmic in $N$.

It should be noted that calculating parities for FEC codes requires
$O(N)$ operations ($N$ multiplications, $N$ additions to add the
results). However, the summation step on a parallel
computer can be reduced to $O(\log(N))$ by requiring $N$ parallel
processors and shared memory (PRAM). The $N$ multiplications can be
performed in parallel, in quasi-logarithmic time
$O(\log(N)\cdot \log\log(N)\cdot\log\log\log(N))$.  Therefore, if
error correction can be performed in a fixed number of Galois
operations (not depending on $N$), the overall algorithm remains
quasi-logarithmic on a parallel computer with $N$ processors.

\section{Error Correcting Capabilities for $3$ Failed Drives}
\label{sec:three_drives}
In this section we obtain results on detecting and correcting of $3$
errors. Clearly, for a code of distance $d=5$ code, one can expect to
be able to correct only $\lfloor(d-1)/2\rfloor=2$ errors by a decoder
which searches for the nearest valid codeword (minimum distance
decoder), such as ours. However, since we have $5$ parities, it turns
out that our code has an advantage over a hypothetical code with
$4$ parities, when it comes to detecting and correcting $3$ failing drives.

The main idea is that of list decoding. When the syndrome vector
$s\in\FF^5$ is determined not to be consistent with $2$ errors (any
combination of data and parity errors), we are able to find all
possible vectors $e$ of weight $3$ such that $s=H\,e$. This strategy
may be successful if the set of possible solutions is not too large,
and that there exists an efficient algorithm to compute this set.

Let us consider the case when $3$ data drives failed first.  We note
that $s_2+s_3+s_5=0$ is a necessary condition for a syndrome vector
$s$ to be due to pure data drive failures. Therefore, we can drop
$s_5$ after checking this condition.
\begin{theorem}[On three failed data drives]
  \label{thm:three-data-drives}
  Let us consider the code introduced by
  equations~\eqref{eqn:generator-matrix},
  \eqref{eqn:parity-check-matrix} and \eqref{eqn:parity-matrix}.  Let
  us suppose that a syndrome vector $s=H\cdot e$ comes from failure of
  exactly $3$ data drives at locations $i$, $j$ and $l$, and there is
  no failure of $2$ drives which results in $s$.  Then $s_2+s_3+s_5=0$
  and
  \begin{equation}
    \label{eqn:3-drive-failure-condition}
    s_4-\sigma_1\,s_3+\sigma_2\,s_2-\sigma_3\, s_1=0
  \end{equation}
  where $\sigma_m$, $m=1,2,3$, are the symmetric polynomials of
  $\alpha_i$, $\alpha_j$ and $\alpha_l$:
  \begin{align*}
    \sigma_1&=\alpha_i+\alpha_j+\alpha_l,\\
    \sigma_2&=\alpha_i\,\alpha_j+\alpha_i\,\alpha_l+\alpha_j\,\alpha_l,\\
    \sigma_3&=\alpha_i\,\alpha_j\,\alpha_l.
  \end{align*}
\end{theorem}
\begin{proof}
  Let us set $a=\alpha_i$, $b=\alpha_j$ and $c=\alpha_l$.  We need to
  study the solutions of the system
  \[ x\cdot\P{a} + y\cdot\P{b} + z\cdot\P{c} = s\]
  which in full form is:
  \begin{align*}
    x+y+z &=s_1,\\
    a\,x+b\,y+c\,z &=s_2,\\
    a^2\,x+b^2\,y+c^2\,z &=s_3,\\
    a^3\,x+b^3\,y+c^3\,z &=s_4,\\
    a(a+1)\,x+b(b+1)\,y+c(c+1)\,z &= s_5.
  \end{align*}
  The last equation, in view of $s_5=s_2+s_3$, can be dropped, as it
  is the sum of the second and third equation.  The first four
  equations form a linear system for $x$, $y$, and $z$.  which is
  overdetermined. Moreover, the first $3$ equations have a coefficient
  matrix which is a $3\times 3$ Vandermonde matrix with determinant
  $D_0=(b-a)\,(c-a)\,(c-b)$. We may assume that $D_0\neq 0$ as $a$, $b$
  and $c$ range over distinct elements of the set $\alpha$. Thus, the
  system is consistent iff the $4\times 4$ augmented coefficient
  matrix has rank $3$.  This determinant can be calculated using CAS,
  and is:
  \[ D=D_0\cdot (s_4-s_3\,\sigma_1+s_2\,\sigma_2-s_1\,\sigma_3) \]
  Hence, the consistency condition is equivalent to equation
  \eqref{eqn:3-drive-failure-condition}.
\end{proof}
Theorem~\ref{thm:three-data-drives} limits the number of triples
$(a,b,c)$ to $k(k-1)/2$, because by trying all combinations
$\{a,b\}\subseteq \alpha$ we find $c$ from
\eqref{eqn:3-drive-failure-condition}. 

In underlying applications to RAID, it may be possible to obtain
several syndromes which allow us to further limit the failed data
drive locations. In fact, most commonly we will have a sample of
several syndrome vectors due to failed drives. We can also obtain a
sample by repeatedly reading and writing suspect stripes of data on
the drives. This approach may quickly succeed. The criterion of
success is given in the next theorem, and it roughly consists in
checking linear independence of syndromes, which is a straightforward
task.
\begin{theorem}
  \label{thm:independent-syndromes}
  Let $s^{(m)}=(s_{j}^{(m)})_{j=1}^5$, $m=1,2,\ldots,M$, be syndrome vectors, i.e.\ any vectors in the range of $H$,
  and 
  \begin{equation*}
    s_{5}^{(m)} = s_{2}^{(m)}+s_{3}^{(m)}\quad\text{for $m=1,2,\ldots,M$}.
  \end{equation*}
  Then the following is true:
  \begin{enumerate}
    \item 
      If $M=3$ and the matrix $(s_{j}^{(m)})_{j,m=1}^3$ is
      non-singular, then the locations $(i,j,l)$ of the failed drives
      which may result in these syndromes can be found by first
      finding $(\sigma_1,\sigma_2,\sigma_3)$ from the linear system:
      \begin{equation*}
        \sum_{j=1}^3 (-1)^{j-1}\,s_{4-j}^{(m)}\,\sigma_j = s_{4}^{(m)},\quad\text{$m=1,2,3$}
      \end{equation*}
      and then solving the cubic equation:
      \begin{equation*}
        (\zeta-a)\,(\zeta-b)\,(\zeta-c)=\zeta^3-\sigma_1\zeta^2+\sigma_2\zeta-\sigma_3=0.
      \end{equation*}
      The solution is the unique triple $(a,b,c)$. We may determine $(i,j,l)$ by
      lookup, equating $\alpha_i=a$, $\alpha_j=b$ and $\alpha_l=c$.
    \item If $M=4$ and the syndrome vectors are linearly independent
      then cannot come from a failure of $\leq 3$ data drives.

    \item If $M=2$ and such syndrome vectors are found such that
      \begin{equation*}
        (s_1^{(1)},s_2^{(1)},s_3^{(1)})\neq (s_1^{(2)},s_2^{(2)},s_3^{(2)})
      \end{equation*}
      where both vectors are non-zero, then the number of triples
      $(i,j,l)$ of failed data drives which can result in those
      syndromes does not exceed $k$.
  \end{enumerate}
\end{theorem}
\begin{proof}
  The case of $3$ syndromes is obvious, in view of our preceding
  analysis. If there are $4$ linearly independent syndromes then there
  is no solution for $(\sigma_1,\sigma_2,\sigma_3)$. If there are $2$
  syndromes are found as described in the theorem then there are $2$
  linearly independent syndromes then the set of triples
  $(\sigma_1,\sigma_2,\sigma_3)$ form a $1$-dimensional affine
  subspace of $\FF^3$. Thus, varying one of the variables $\sigma_1$,
  $\sigma_2$ or $\sigma_3$ (the free variable) over the set $\alpha$
  and finding the other two from the system of linear equations,
  yields not more than $k$ solutions.
\end{proof}
\begin{remark}[On locating and recovery of $3$ failed data drives]
  Based on Theorem~\ref{thm:independent-syndromes} we have a clear
  strategy to locate and correct $3$ failed data drives. We simply
  collect syndromes and observe their projections onto the first
  $3$ coordinates. Once we find $3$ independent vectors in our collection,
  we can locate the failed drives. We can clearly correct the resulting
  errors, as we can correct up to $4$ errors at known locations.
\end{remark}

The relevant bound for all other cases is the subject of our next theorem.
\begin{theorem}[On number of solutions for $3$-drive failure]
  \label{thm:3-drive-list-decoding}
  Let $s\in\FF^5$ be a vector such that there is no vector
  $e\in\FF^{k+5}$ of weight $2$ for which $H\,e=s$. Then the number of
  triples $(i,j,l)$, $1\leq i<j<l \leq k$, such that there is a vector $e$
  of weight $3$, such that
  \begin{enumerate}
    \item $s$ is a syndrome vector for $e$, i.e.\ $H\,e=s$;
    \item $e_m=0$  unless $m\in\{i,j,l\}$ ($1\leq m\leq k+5$);
    \item $l>k$, i.e.\ not all three failed drives are data drives;
  \end{enumerate}
  is not more than $2\,k+4$. 
  For given $s$ and triple $(i,j,l)$, the vector $e$ with
  the above properties is unique.
\end{theorem}
\begin{proof}
  Cases of failure of 3 drives can be divided according to the
  number of failed parity drives. 

  If $3$ parity drives fail, at positions $k+i$, $k+j$, $k+m$, where
  $1\leq~i<j<m\leq 5$ the equation is
  \[ x\cdot I_i + y\cdot I_j + z\cdot I_m=s.\]
  This equation implies that $s$ has weight $3$ and syndromes $i$, $j$
  and $m$ are the non-zero syndromes. Moreover, $x=s_i$, $y=s_j$ and
  $z=s_m$. Thus, for every $s$ of weight $3$ there exists a unique
  solution of this type. The error vector satisfies $e_{k+i}=x$,
  $e_{k+j}=y$ and $e_{k+m}=z$, and $e_l=0$ for $l\notin\{i,j,m\}$.

\input Table1Data2Parity.tex

  If $2$ parity drives fail, at positions $k+j$, $k+l$, where
  $1\leq~j<l\leq 5$, along with data drive $i$ then we have
  \[ x\cdot\P{\rho}+y\cdot I_j+z\cdot I_m =s \]
  where $\rho=\alpha_i$. Let $\P{\rho}_{j,m}$ be the vector $\P{\rho}$ with
  entries $j$ and $l$ erased. Thus
  \[ x\cdot\P{\rho}_{j,l}=s^{(j,l)} \]
  where $s^{(j,l)}$ is the syndrome vector $s$ with entries $j$ and
  $m$ erased.  This gives us $3$ equations for $x$ and $\rho$, which
  should allow us to solve the
  problem. Table \ref{tab:fixing-1-data-2-parity} contains the results of
  careful analysis of all cases for distinct pairs $(j,l)$. As we can
  see, in each case we have multiple (two or three) equations which
  $\rho$ satisfies, with coefficients dependent on the syndromes
  $s_t$, $t=1,2,\ldots,5$ (column $3$ of the table). We also can
  obtain constraints on the syndromes by eliminating $\rho$ from the
  equations in column $2$. These are listed in column $3$. As we can
  see, in each case there is exactly one constraint.  A lengthy
  analysis shows that $\rho$ is unique, except for the degenerate
  situation, when $s$ is a syndrome vector for an error vector of
  weight $\leq 2$. The arguments are straightforward but lengthy, and
  are omitted. We only mention the case $j=2$ and $l=3$ as it is
  different from other cases in one respect, that it relies upon the
  exclusion rule for cubic roots of unity. If $\rho$ is non-unique
  then $s_5+s_1=0$, $s_5+s_4=0$, $s_5+s_4\,s_5=0$, $s_4^2+s_1\,s_4=0$
  and $s_5^3+s_1\,s_4\,s_5+s_1\,s_4^2+s_1^2\,s_4=0$. These equations
  imply $s_1=s_4$. Also $s_5\,(s_5+s_4)=0$. If $s_5=0$ then also $s_1=0$
  and $s_4=0$, Thus $s$ has weight $2$, which is a syndrome vector for
  $\leq 2$ failed parity drives with numbers in the set
  $\{2,3\}$. Thus, we may assume $s_5\neq 0$. Then $s_5=s_4$ and thus
  $s_1=s_4=s_5$. The second equation in column $2$ reduces to
  $s_5(\rho^2+\rho+1)=0$, which implies $\rho^2+\rho+1=0$. Therefore
  $\rho$ is a cubic root of unity $\neq 1$. It must therefore be the
  root of unity different from the excluded one. This makes $\rho$
  unique.

  Hence, there is at most one solution with $2$ failed data and $1$
  failed parity drive.

  If $1$ parity drive fails at position $j$, along with 2 data drives,
  we have
  \begin{equation*}
    x\cdot\P{a}+ y\cdot\P{b}+z\cdot I_j=s. 
  \end{equation*}

\begin{table}[htb]
  \caption{\label{tab:fixing-2-data-1-parity} Systems of equations for
    recovery from a failure of $3$ drives, two data and $1$ parity at
    location $k+j$, $j=1,2,3,4,5$. The equations are expressed in
    terms of symmetric polynomials  $\sigma_1=a+b$ and $\sigma_2=a\cdot b$.}
  \begin{center}
      \begin{tabular}{||c||l|l||}
        \hline\hline
        $j$ & System of equations for $\sigma_1$ and $\sigma_2$ & Constraint on $s$\\
\hline\hline
\multirow{3}{*}{5} & $  s_{3}+\sigma_{1}\,s_{2}+s_{1}\,\sigma_{2} $ &\\
& $s_{4}+\sigma_{1}\,s_{3}+\sigma_{2}\,s_{2} $ &\\
& $s_{2}\,s_{4}+s_{3}^2+\sigma_{2}\,\left(s_{1}\,s_{3}+s_{2}^2\right) $ &\\ 
\hline
\multirow{1}{*}{4} & $\sigma_{1}\,\left(s_{5}+s_{3}\right)+s_{3}+s_{1}\,\sigma_{2} $ & $  s_{5}+s_{3}+s_{2} $ \\ 
\hline
\multirow{3}{*}{3} & $  s_{5}+\sigma_{1}\,s_{2}+s_{2}+s_{1}\,\sigma_{2} $&\\
& $\sigma_{1}\,s_{5}+s_{5}+s_{4}+\sigma_{2}\,\left(s_{2}+s_{1}\right)+s_{2} $&\\
&$s_{5}^2+\sigma_{2}\,\left(s_{1}\,s_{5}+s_{2}^2+s_{1}\,s_{2}\right)+s_{2}\,s_{4}+s_{2}^2 $& \\ 
\hline
\multirow{3}{*}{2} & $ \sigma_{2}\,\left(s_{5}+s_{3}+s_{1}\right)+\sigma_{1}\,s_{5}+s_{4}+s_{3} $ &\\
&$\sigma_{2}\,\left(s_{5}^2+s_{3}\,\left(s_{3}+s_{1}\right)\right)+s_{4}\,s_{5}+s_{3}\,s_{4}+s_{3}^2 $ &\\
&$\sigma_{2}\,\left(s_{5}+s_{3}\right)+s_{4}+\sigma_{1}\,s_{3} $ &\\ 
\hline
\multirow{1}{*}{1} & $\sigma_{2}\,\left(s_{5}+s_{3}\right)+s_{4}+\sigma_{1}\,s_{3}  $ & $s_{5}+s_{3}+s_{2} $\\
\hline\hline
      \end{tabular}
  \end{center}
\end{table}

  For fixed $j$, this is a system of equations for $x$, $y$, $a$, $b$
  and $z$, i.e.\ $5$ equations in $5$ unknowns, i.e.\ the problem is
  well-posed.  The case is thus subdivided into subcases according to
  the value of $j$. We preprocessed the equations with CAS, by first
  erasing equation in row $j$ (which eliminates $z$), and then eliminating variables $x$ and $y$.
  Also, since the system is symmetric with respect to $a$ and $b$, we expressed the equations
  in terms of symmetric polynomials $\sigma_1=a+b$ and $\sigma_2=a\cdot b$. The
  result is in Table~\ref{tab:fixing-2-data-1-parity}. It should be noted
  that in each case we have a linear system of equations for $(\sigma_1,\sigma_2)$.
  Each solution of the linear system yields a single solution $(a,b)$ up to swapping $a$ and $b$.
  We proceed to more precisely determine the number of solutions.
  The analysis of subcases for $j=5,4,3,2,1$ is as follows:
  \subsection*{Case $j=5$}
  There are 3 linear equations for $(\sigma_1,\sigma_2)$. The solution of the system is non-unique iff
  $s_1\,s_3+s_2^2=0$ and $s_2\,s_4+s_3^3=0$. This system of equations can also be written as 
  \begin{align*}
    s_1\,s_3 &= s_2^2,\\
    s_2\,s_4 &= s_3^2.
  \end{align*}
  If $s_2=0$ then $s_3=0$, and $s_4=0$. Thus $s$ has weight at most
  $2$, and it matches the case of $2$ failed parity drives, which is a
  contradiction. Therefore, $s_2\neq 0$ and $s_1\neq 0$, $s_3\neq 0$,
  $s_4=s_3^2/s_2\neq 0$. Let us define $\rho=s_2/s_1\neq 0$.  We have
  $s_2=s_1\,\rho$, $s_3=s_2^2/s_1=s_1\rho^2$,
  $s_4=s_3^2/s_2=s_1\,\rho^3$.  Hence, $s=x\,\P{\rho}+z\,I_5$, where
  $x=s_1$, for some $z$. This is also a contradiction. Hence, we
  may assume that the linear system for $j=5$ is non-singular.
  Hence, there is a unique solution $(\sigma_1,\sigma_2)$. 

  We conclude that there exists a unique solution with $j=5$.

  \subsection*{Case $j=4$} 
  The first equation, $s_2+s_3+s_5=0$, is a necessary condition on the
  syndromes for this case to be possible. The second equation yields a
  relation between $a$ and $b$, more precisely, a linear relationship
  between the symmetric polynomials $\sigma_1$ and $\sigma_2$, which
  can be rewritten as $s_1\,\sigma_2+s_2\,\sigma_1+s_3=0$. Unless
  $s_1=s_2=0$ this constraint is non-degenerate. If $s_1=s_2=0$, also
  $s_3=0$. Hence $s_5=s_2+s_3=0$. Hence, $s_4$ can be the only
  non-zero syndrome. In this case $a$ and $b$ are arbitrary. However,
  $x$ and $y$ are determined to be $0$, so no data drives have
  failed. Thus, in contradiction with our assumption, there is only
  one failed drive: parity drive at position $k+4$. Hence, we may
  assume that either $s_1\neq 0$ or $s_2\neq 0$. By letting $b$ assume
  all $k$ values $\alpha_1$, $\alpha_2$, $\ldots$, $\alpha_k$, we
  determine $a$ from the linear equation $(s_1\,b+s_2)\,a=s_3+s_2\cdot
  b$. If $s_1=0$, $a=s_3/s_2+b$. If $s_1\neq 0$ then
  $a=(s_3+s_2\,b)/(s_1\,b+s_2)$ is unique for $b\neq s_3/s_2$. If
  $b=s_3/s_2$ then $a$ is $\neq b$ and otherwise arbitrary.  This
  leads to $2(k-1)$ pairs $(a,b)$.  Also, there is a symmetry: if
  $(a,b)$ is a solution, so is $(b,a)$. This symmetry shows that every
  solution is repeated twice in the above procedure.

  Thus the  number of solutions for $j=4$ is bounded by $k-1$ in total.

  \subsection*{Case $j=3$}
  Condition of non-unique solution is that $s_1\,s_5+s_1\,s_2+s_2^2=0$, or $s_1(s_2+s_5)=s_2^2$,
  and $s_2\,s_4+s_2^2+s_5^2=0$, or $s_2\,(s_2+s_4)=s_5^2$. Thus,
  \begin{align*}
    s_1(s_2+s_5)&=s_2^2,\\
    s_2\,(s_2+s_4)=s_5^2.
  \end{align*}
  Let us suppose that $s_2\neq 0$. Then $s_1\neq 0$ and $s_2+s_5\neq
  0$. Let us define $\rho=s_2/s_1\neq 0$. Then
  $s_1(s_1\rho+s_5)=s_1^2\rho^2$.  Hence
  $s_1\,s_5=s_1^2(\rho^2+\rho)$, or
  $s_5=s_1\,\rho\,(\rho+1)$. Furthermore,
  \begin{equation*}
    s_4 = s_5^2/s_2 - s_2 = s_1^2\,\rho^2\,(\rho^2+1)/(s_1\,\rho)-s_1\,\rho 
    = s_1\,(\rho^3+\rho)-s_1\,\rho=s_1\,\rho^3.
  \end{equation*}
  We thus have proven
  $s=s_1(1,\rho,s_3/s_1,\rho^3,s_1\rho(\rho+1))$. This implies
  $s=x\P{\rho}+z\,I_3$ where $x=s_1$ and $z=s_3-x\rho^3$. This means
  that $s$ matches a solution with just two drives failed, which is a
  contradiction.

  Let us suppose that $s_2=0$. Then $s_1s_5=0$ and $s_1\,s_4=0$.  If
  $s_1\neq 0$ then $s_2=s_5=s_4=0$ and $s$ has weight $\leq 2$ which
  is consistent with $2$ parity drive failure. So $s_1=s_2=0$. Also
  $s_5=0$ and $s_4=0$.  Thus $s=0$, which is consistent with no drive
  failing, which is again a contradiction.

  The solution for $j=3$ is therefore unique.

  \subsection*{Case $j=2$}
  The linear
  system for $\sigma_1$ and $\sigma_2$ is singular iff 
  $s_1\,s_2-(s_3+s_5)^2=0$. Moreover, the coefficients of the last
  equation simultaneously vanish iff 
  $s_5^2+s_3\,(s_3+s_1)=0$ and $s_4\,s_5+s_3\,s_4+s_3^2=0$.
  This system is equivalent to 
  \begin{align*}
    s_3\,(s_1+s_3)=s_5^2,\\
    s_3\,(s_3+s_4)=s_4\,s_5.
  \end{align*}
  Let us suppose $s_5\neq 0$. Then $s_3\neq 0$ and $s_1+s_3\neq 0$. If
  $s_4=0$ then the second equation yields $s_3^2=0$, which implies
  $s_3=s_4=0$. Hence, $\sigma_2\,s_5=0$ by the third equation in
  Table~\ref{tab:fixing-2-data-1-parity}.  But $\sigma_2\neq 0$, as
  only $a,b\neq 0$ are solutions. Therefore $s_5=0$, which contradicts
  our assumption that $s_5\neq 0$.  Hence $s_4\neq 0$. This implies
  $s_3\neq 0$ and $s_3+s_4\neq 0$.  If $s_1=0$ then the first equation
  implies $s_3^2=s_5^2$, and thus $s_3=s_5$ (we used Frobenius
  identity).  Last second equation yields $s_3+s_4=s_4$, i.e.\ $s_3=0$,
  which would be a contradiction. Hence, $s_1\neq 0$. This implies
  that $s_1,s_3,s_4,s_5\neq 0$. Let $\rho=s_4/s_3\neq0$. Then.
  $s_3\,(s_3+s_3\,\rho)=s_3\rho\,s_5$. Thus $s_5=s_3\,(1+\rho^{-1})$.
  Also $s_3(s_1+s_3)=s_5^2=s_3^2(1+\rho^{-2})$. Hence,
  $s_1+s_3=s_3(1+\rho^{-2})$ and $s_1=s_3\rho^{-2}$. Thus
  $s_3=s_1\,\rho^2$. Hence
  $s=s_1(1,s_2/s_1,\rho^2,\rho^3,\rho(\rho+1)$. Again, there is a
  solution to $x\P{\rho}+z\,I_2=0$, which is a contradiction.

  Let us suppose $s_5=0$. But then the system becomes
  \begin{align*}
    s_3^2&=s_1\,s_3,\\
    s_3\,s_4&=s_3^2.
  \end{align*}
  If $s_3\neq 0$ (in addition to $s_5\neq 0$) then $s_1,s_4\neq 0$.
  Again, we define $\rho=s_4/s_3\neq 0$ and obtain
  $s_1=s_3^2/s_4=s_3/\rho^2$.  So
  $s=s_1(1,s_2/s_1,\rho^2,\rho^3,0)$. But also $s_3=s_1$ which implies
  $\rho^2=1$.  Hence, $\rho=1$. This implies that $s=x\P{1}+z\,I_2$
  has a solution, which is a contradiction.

  If $s_3=0$ (and $s_5=0$ by assumption) then $s_4=0$
  by the third equation in Table~\ref{tab:fixing-2-data-1-parity}.
  But then $s$ has weight at most $2$, and it matches $2$ failed
  parity drives, which contradicts our assumptions.

  We proved that all degenerate cases come from syndrome vectors which match
  failure of fewer than $3$ drives.

  The solution for $j=2$ is therefore unique.  

  \subsection*{Case $j=1$}
  In this case, equation $s_5=s_2+s_3$ is required for the solution to
  exist. The second equation yields a linear constraint on
  $(\sigma_1,\sigma_2)$:
  \[ s_2\,\sigma_2 + s_3\,\sigma_1+s_4 =0. \]
  This constraint is consistent and non-trivial unless
  $s_2=s_3=s_4=0$. In this case, also $s_5=s_2+s_3=0$.  Hence
  $s=s_1(1,0,0,0,0)$. But this means $s=s_1\, I_1$, i.e.\ $s$ is
  consistent with one parity drive failure, which is a
  contradiction.

  Therefore, we may assume that the constraint is non-trivial and the
  set of admissible pairs $(\sigma_1,\sigma_2)$ forms a
  $1$-dimensional linear subspace of $\FF^2$.  Moreover, since
  $(\zeta-a)\,(\zeta-b)=\zeta^2-\sigma_1\zeta+\sigma_2=0$, for every
  admissible pair $(\sigma_1,\sigma_2)$ there are at most $1$
  solutions $(a,b)$, up to swapping $a$ and $b$. Hence, the total
  number of solutions is $\leq\cdot|\FF|$.  A sharper estimate is
  obtained by considering setting $b=\alpha_i$, $i=1,2,\ldots,k$. For
  fixed $b$, the constraint is $s_2\,(a\,b)+s_3\,(a+b)+s_4=0$ or
  $(s_2\,b+s_3)\,a+s_4=0$. If $b=s_3/s_2$ and $s_4=0$ then $a$ is
  $\neq b$ and otherwise arbitrary, yielding $k-1$ possible
  solutions. If $b\neq s_3/s_2$, $a=s_4/(s_2\,b+s_3)$ is
  unique. Therefore, the number of pairs $(a,b)$ of this type is again
  $k-1$. Again, due to symmetry, we can eliminate half the pairs $(a,b)$ 
  The total number of solutions of $j=1$ is thus not greater than $k-1$.

  Finally, the total number of solutions with fixed $s$, which may
  match all of the cases, is bounded by:
  \[ 2\,(k-1)+1 + 5 + 3\cdot 1= 2\,k+4. \]
\end{proof}

The techniques of this proof are easily implemented as a collection of 
algorithms. The top-level algorithm is
Algorithm~\ref{alg:recover-three-failed-disks}. This algorithm can be
invoked after trying all $2$-disk failures against the syndrome vector
$s$, to create a list of all matching error vectors of weight $\leq
3$, except for a simultaneous failure of $3$ data drives.

\begin{algorithm}[htb]
  \caption{ An algorithm which produces the augmented matrix
    $C=[A\,|\,b]$ of the linear system $A\cdot \sigma=b$, where
    $\sigma=(\sigma_1,\sigma_2)$ and $\sigma_1=u+v$, $\sigma_2=u\,v$
    are symmetric polynomials of $u$ and $v$. Galois field element
    $cond$ represents the extra constraint value that must be 0 for
    the system to be consistent.  This algorithm is based on
    Table~\ref{tab:fixing-2-data-1-parity} Matrix $C$ has either 1 or
    2 rows.}
  \begin{algorithmic}[1]
    \Function{AugmatrixForParity}{$s$,$j$}
    \If{$j=5$}
    \State $C\gets\begin{pmatrix}
           s_2 & s_1 & s_3\\
           s_3 & s_2 & s_4
         \end{pmatrix}$
    \State $cond\gets 0$
    \ElsIf{$j=4$}
    \State $C\gets\begin{pmatrix}
           s_2 & s_1 & s_3\\
         \end{pmatrix}$
    \State $cond\gets s_2+s_3+s_5$
    \ElsIf{$j=3$}
    \State $C\gets\begin{pmatrix}
           s_2 & s_1 &       s_2+s_5\\
           s_5 & s_1+s_2 &  s_2+s_4+s_5
         \end{pmatrix}$
    \State $cond\gets 0$;
    \ElsIf{$j=2$}
    \State $C\gets\begin{pmatrix}
           s_5 & s_1+s_3+s_5 & s_3+s_4\\
           s_3 & s_3+s_5 &      s_4
         \end{pmatrix}$
    \State $cond\gets 0$
    \ElsIf{$j=1$}
    \State $C\gets\begin{pmatrix} s_3 & s_2 & s_4 \end{pmatrix}$
    \State $cond\gets s_2+s_3+s_5$
    \EndIf
    \State $\Return{(C,cont)}$
    \EndFunction
  \end{algorithmic}
\end{algorithm}
\begin{algorithm}[htb]
  \caption{\label{alg:locate-two-data-for-parity}Solve the equation
    $x\,\P{u}+y\,\P{v}+z\,I_j=s$. Return all solutions in vectors $u$,
    $v$, $x$, $y$ and $z$. Also, return the solution count $cnt$. If $cnt=0$ then
    $u=v=x=y=z=[]$ (empty vector).}
  \begin{algorithmic}[1]
    \Function{LocateTwoDataForParity}{$\alpha$,$s$,$j$}
    \State $u\gets[]$; $v\gets[]$; \Comment{Initialize to empty vectors.}
    \State $x\gets[]$; $y\gets[]$; $z\gets[]$\Comment{Initialize to empty vectors.}
    \State $cnt\gets 0$
    \State $(C,cond)\gets\Call{AugmatrixForParity}{s,j}$
    \State $cnt\gets 0$;
    \If{$cond\neq 0 $}
    \State \Goto{alg:locate-two-data-for-parity:end}
    \EndIf
    \If{$\Call{NumberOfRows}{C}=2$}
    \State $\{u,v\}\gets\Call{LocateTwoDataWhenDetermined}{C}$
    \If{$u=0 \And v=0$}
    \State \Goto{alg:locate-two-data-for-parity:end}
    \EndIf
    \State $(x,y,z)\gets\Call{CalculateCoefficients}{u,v,s,j}$;
    \State $cnt\gets 1$

    \ElsIf {$\Call{NumberOfRows}{C}=1$}
    \State $(u,v,x,y,z,cnt)\gets\Call{LocateTwoDataWhenUnderdetermined}{\alpha,s,j,C}$
    \EndIf
    \State \Return{$(u,v,x,y,z,cnt)$}
    \label{alg:locate-two-data-for-parity:end} \Comment{No solution exists}
    \EndFunction
  \end{algorithmic}
\end{algorithm}
\begin{algorithm}[htb]
  \caption{\label{alg:locate-two-data-when-determined}An algorithm implementing a helper function for Algorithm~\ref{alg:locate-two-data-for-parity}. This simple algorithm first solves a $2\times 2$ linear system by Cramer's Rule and then solves a quadratic equation to find $u$ and $v$.}
  \begin{algorithmic}[1]
  \Function{LocateTwoDataWhenDetermined}{$C$}
  \State $u\gets 0$; $v\gets 0$ \Comment{Zero in Galois field.}
  \State $D\gets C_{11}\,C_{22}-C_{21}\,C_{12}$
  \State $D_1\gets C_{13}\,C_{22}-C_{23}\,C_{12}$
  \State $D_2\gets C_{11}\,C_{23}-C_{21}\,C_{13}$
  \If{$D=0 \Or D_1=0 \Or D_2=0$}
  \State \Goto{alg:locate-two-data-when-determined:end}
  \EndIf
  \State $\sigma_1\gets D_1/D$
  \State $\sigma_2\gets D_2/D$
  \State $roots\gets\Call{SolveQuadraticEquation}{1,-\sigma_1,\sigma_2}$
  \If{$roots=\emptyset$}
  \State \Goto{alg:locate-two-data-when-determined:end}
  \Else \Comment{In view of $\sigma_1\neq 0$, there are two distinct roots.}
  \State $u\gets roots_1$
  \State $v\gets roots_2$
  \EndIf
  \State \Return{$(u,v)$}\label{alg:locate-two-data-when-determined:end}
  \EndFunction
  \end{algorithmic}
\end{algorithm}
\begin{algorithm}[htb]
  \caption{\label{alg:locate-two-data-when-underdetermined}
    If $C=[A|b]$ and $A=[c_1,c_2]$, $b=c_3$ then the system $A\,\sigma=b$ has only one equation
    $c_1\,\sigma_1+c_2\,\sigma_2=c_3$. We find
    $c_1\,(u+v)+c_2\,u\,v=c_3$,
    $(c_1+c_2\,v)\,u=c_3-c_1\,v$,
    $u=(c_3-c_1\,v)/(c_1+c_2\,v)$.
    The above formula yields unique $u$ when $v\neq -c_1/c_2$.
    Otherwise, $u$ is arbitrary if 
    $c_3-c_1\,(-c_1/c_2)=c_3-c_1^2/c_2=0$.
    Thus, $u$ is arbitrary if $c_3\,c_2=c_1^2$.}
  \begin{algorithmic}[1]
    \Function{LocateTwoDataWhenUnderdetermined}{$\alpha$,$s$,$j$,$c$}
    \State $cnt\gets 0$
    \State $k\gets\Call{NumberOfElements}{\alpha}$
    \For{$l=1,2,\ldots,k$}
    \State $v\gets\alpha_l$
    \State $N\gets c_3-c_1\,v$
    \State $D=c_1+c_2\,v$
    \If{$D\neq 0 \And N\neq 0$}
    \State $u\gets N/D$
    \If{$u<v$}                      \Comment{Sort pairs to avoid duplicates.}
    \State {$(x,y,z)=\Call{CalculateCoefficients}{u,v,s,j}$};
    \State $cnt\gets cnt+1$
    \State $uLst_{cnt}=u$;$vLst_{cnt}=v$; $xLst_{cnt}=x$; $yLst_{cnt}=y$; $zLst_{cnt}=z$
    \EndIf
    \ElsIf{$N=0 \And D=0$}
    \Comment{Now u is arbitrary, the only constraint is $u\neq v$.}
    \For{$u\in\alpha$}
    \If{$u<v$}                  \Comment{ Sort pairs to avoid duplicates.}
    \State $(x,y,z)=\Call{CalculateCoefficients}{u,v,s,j}$                
    \State $cnt\gets cnt+1$
    \State $uLst_{cnt}=u$;$vLst_{cnt}=v$; $xLst_{cnt}=x$; $yLst_{cnt}=y$; $zLst_{cnt}=z$
    \EndIf
    \EndFor
    \EndIf
    \EndFor
    \EndFunction
  \end{algorithmic}
\end{algorithm}

\begin{algorithm}[htb]
  \caption{\label{alg:calculate-coefficients} This helper algorithm
    implements function \textproc{CalculateCoefficients} which finds the
    coefficients $x$, $y$ and $z$ for two data and one parity error.
    A call $(x,y,z)\gets\textproc{CalculateCoefficients}(u,v,s,j)$ solves
    the vector equation in Galois field: $x\,\P{u}+y\,\P{v}+z\,I_j=s$
    The arguments $u$ and $v$ must be non-zero and distinct, and $J$
    must be in the range $1\leq j\leq 5$.}
  \begin{algorithmic}[1]
    \Function{CalculateCoefficients}{$u$,$v$,$s$,$j$}
    \If{$j>2$}
    \Comment{Use rows $1$ and $2$.}
    \State $D\gets v-u$
    \State $x\gets (s_1\,v-s_2)/D$
    \State $y=(s_1\,u-s_2)/D$
    \ElsIf{$j=2$}          \Comment{Use rows $1$ and $3$.}
    \State $D\gets(v-u)^2$
    \State $x\gets (s_1\,v^2-s_3)/D$
    \State $y\gets (s_1\,u^2-s_3)/D$
    \ElsIf{$j=1$}          \Comment{Use rows $2$ and $3$.}
    \State $D\gets u\,v\,(v-u)$
    \State $x\gets(s_2\,v^2-s_3\,v)/D$
    \State $y\gets(s_2\,u^2-s_3\,u)/D$
    \EndIf
    \If{$j<5$}
    \State $z\gets s_j-x\,u^{j-1}-y\,v^{j-1}$
    \Else                          \Comment{$j=5$}
    \State $z\gets s_j-x\,u\,(u+1)-y\,v(v+1)$
    \EndIf
    \EndFunction
  \end{algorithmic}
\end{algorithm}

\begin{algorithm}[htb]
  \caption{\label{alg:recover-data-for-two-parity} 
    Implements recovery of three failed drives, case of 2 parity and 1
    data error. It accepts as arguments the set $\alpha\subset\FF $, the syndrome vector $s\in\FF^5$,
    the index of the broken parities $j$  and $l$ ($1\leq j<l \leq 5$), and the set of excluded cubic roots of unity $X$. It returns
    the error vector $e\in\FF^{k+5}$.}
  \begin{algorithmic}[1]
    \Function{RecoverDataForTwoParity}{$\alpha$,$s$,$j$,$l$,$X$}
    \State $k\gets\Call{NumberOfElements}{\alpha}$; $e\gets []$; $i\gets []$; $\rho\gets 0$ \Comment{$0\in\FF$}
    \If{$j\gets 1$ }
    \If{$l=2 \And s_4\,s_5+s_3\,s_4+s_3^2=0$}
    \If{ $s_3\neq0$}
    \State $\rho\gets s_4/s_3$
    \ElsIf{ $s_5\neq0$ }
    \State $\rho\gets (s_3+s_4)/s_5$
    \EndIf
    \ElsIf{ $ l=3 \And s_5^2+s_2\,s_4+s_2^2=0 $}
    \If{$s_2\neq0 $}
    \State $\rho\gets (s_2+s_5)/s_2$
    \ElsIf{$s_5\neq0$}
    \State $\rho\gets (s_2+s_5+s_5)/s_5$;
    \EndIf
    \ElsIf{$l=4 \And s_2+s_3+s_5=0$}
    \If{$s_2\neq0 $}
    \State $\rho\gets s_3/s_2$
    \EndIf
    \ElsIf{$l=5 \And s_2\neq0 \And s_2\,s_4+s_3^2=0$}
    \State $\rho\gets s_3/s_2$
    \EndIf
    \ElsIf{$j=2 $}
    \If{$l=3 \And (s_5^2+s_1\,s_4)\,s_5+s_1\,s_4\,(s_1+s_4)=0$}
    \If{$s_1+s_5\neq0$}
    \State $\rho\gets (s_4+s_5)/(s_1+s_5)$
    \ElsIf{$ s_5\neq0$}                  \Comment{$s_1=s_5=s_4$ and $\rho^2+\rho+1=0$}
    \If{$X\neq\emptyset$}
    \State $\rho\gets 1-X_1$
    \EndIf
    \EndIf
    \ElsIf{$ l=4 \And (s_3+s_5)^2+s_1\,s_3=0$}
    \If{$s_3+s_5\neq0$}
    \State $\rho\gets s_3/(s_3+s_5)$
    \EndIf
    \ElsIf{$l=5 \And s_1\,s_4^2+s_3^3=0$}
    \If{$s_3\neq0$}
    \State $\rho\gets s_4/s_3$
    \EndIf
    \EndIf
    \algstore{alg-recover-data-for-two-parity-saved}
  \end{algorithmic}
\end{algorithm}

\begin{algorithm}[htb]
  \caption{Algorithm~\ref{alg:recover-data-for-two-parity}, part 2}
  \begin{algorithmic}[1]
    \algrestore{alg-recover-data-for-two-parity-saved}
    \ElsIf{$j=3$}
    \If{$l=4 \And s_1\,(s_5+s_2)+s_2^2=0$}
    \If{$s_2\neq0$}
    \State $\rho\gets (s_2+s_5)/s_2$
    \EndIf
    \ElsIf{$ l=5 \And s_1^2\,s_4+s_2^3=0$}
    \If{$s_2\neq0$}
    \State $\rho\gets s_1\,s_4/s_2^2$
    \EndIf
    \EndIf
    \ElsIf{$j=4$}
    \If{$l=5 \And s_1\,s_3+s_2^2=0$}
    \If{$s_2\neq0$}
    \State $\rho\gets s_3/s_2$
    \EndIf
    \EndIf
    \EndIf
    \If{$\rho=0$}
    \State \Return
    \EndIf
    \State $i\gets\Call{Lookup}{\alpha,\rho}$
    \If{$i\neq\emptyset$}
    \State $p\gets \P{\rho}$
    \State $\{q,t,w\}\gets \{1,2,3,4,5\}\setminus\{j,l\}$ \Comment{The sorted complement of $\{j,l\}$, $q<t<w$.}
    \State $s'\gets (s_{q},s_{t},s_{w})$
    \State $p'\gets (p_{q},p_{t},p_{w})$
    \State $x\gets s_q/p_q$     \Comment{Prospective $x$. Note that $q<5$ implies $p_q\neq 0$, as $\rho\neq 0$.}
    \If{$s_t=x\,p_t \And s_w=x\,p_w$} \Comment{We indeed have a solution.} 
    \State $e\gets 0$          \Comment{$0\in\FF^{k+5}$}
    \State $e_i\gets x$
    \State $e_{k+j}\gets s_j-x\,p_j$
    \State $e_{k+l}\gets s_l-x\,p_l$
    \EndIf
    \EndIf
    \EndFunction
  \end{algorithmic}
\end{algorithm}

\begin{algorithm}[htb]
  \caption{\label{alg:recover-two-data-for-parity} Implements recovery
    of three failed drives, case of 1 parity and 2 data
    errors. Accepts as arguments the set $\alpha\subset\FF$, the
    syndrome vector $s\in\FF^5$, and the broken parity drive index $j$
    ($1\leq j\leq 5$). It returns the error matrix $e$ which has size
    $(k+5)\times cnt$, where $cnt$ is the count of possible error
    vectors which result in $s$. Thus, $H\cdot e=[s|s|\cdots|s]$ where
    $s$ is repeated $cnt$ times. We note that $cnt=0$ is possible,
    when there is no solution meeting our specification.}
  \begin{algorithmic}[1]
    \Function{RecoverTwoDataForParity}{$\alpha$,$s$,$j$}
    \State $(u,v,x,y,z,cnt)\gets\Call{LocateTwoDataForParity}{\alpha,s,j}$
    \State $k\gets\Call{NumberOfElements}{\alpha}$
    \State $e=0$  \Comment{A $(k+5)\times cnt$ matrix of $0\in\FF$.}
    \State $t\gets 0$
    \For{$l=1,2,\ldots,cnt$}
    \State $i_1\gets\Call{Lookup}{\alpha,u_l}$
    \State $i_2\gets\Call{Lookup}{\alpha,v_l}$
    \If{ $i_1=\emptyset \Or i_2=\emptyset$}
    \State \Continue;
    \EndIf
    \State $t\gets t+1$
    \State $e_{i_1,t}\gets x_l$
    \State $e_{i_2,t}\gets y_l$
    \State $e_{k+j,t}\gets z_l$
    \EndFor
    \State $e\gets\Call{SelectColumns}{e,1,t}$\Comment{Keep only columns $1$--$t$.}
    \State \Return{$e$}
    \EndFunction
  \end{algorithmic}
\end{algorithm}

\begin{algorithm}
  \caption{\label{alg:recover-three-failed-disks} Implements recovery
    for $3$-drive failure where at least one of the drives is a parity
    drive. The arguments are: the set $\alpha\subset\FF$, the syndrome
    vector $s\in\FF^5$ and the set of excluded cubic roots of unity,
    $X$ (with $0$ or $1$ element). The result is a matrix $e$ whose
    columns are all possible error vectors satisfying our
    specification. Hence, this algorithm is a partial list decoder for
    all $3$-drive errors, with a notable exception when the $3$ drives
    are data drives.}
  \begin{algorithmic}[1]
    \Function{RecoverThreeFailedDisks}{$\alpha$,$s$,$X$}
    \State $lst\gets\Call{FindNonZeros}{s}$
    \State $nz\gets\Call{NumberOfElements}{lst}$
    \State $k\gets\Call{NumberOfElements}{\alpha}$;
    \State $e\gets[]$
    
    \If{$nz=3$}       \Comment{Find a solution with 3 parity errors.}
    \State $f\gets 0$ \Comment{Vector of $k+5$ zeros in $\FF$.}
    \For{$t\in lst$}
    \State $f_{k+t}\gets s_t$
    \EndFor
    \State $e\gets\Call{AppendColumns}{e,f}$
    \EndIf

    \For{$j=1,2,3,4$} \Comment{Find all solutions with 2 parity errors.}
    \For{$l=j+1,j+2,\ldots,5$}
    \State $f\gets\Call{RecoverDataForTwoParity}{\alpha,s,j,l,X}$
    \State $e\gets\Call{AppendColumns}{e,f}$
    \EndFor
    \EndFor
    \State\Return{$e$}
    
    \For{$j=1,2,3,4,5$} \Comment{Find all solutions with 1 parity error.}
    \State $f\gets\Call{RecoverTwoDataForParity}{\alpha,s,j}$
    \State $e\gets\Call{AppendColumns}{e,f}$
    \EndFor
    \EndFunction
  \end{algorithmic}
\end{algorithm}

\begin{remark}
How can Theorem~\ref{thm:3-drive-list-decoding} be applied?

Recovery of data based on Theorem~\ref{thm:3-drive-list-decoding} will
vary depending on the drive characteristics (note that ``drive'' is
used in a broad sense, to mean any kind of storage device or even a
communications channel).
\end{remark}

\begin{example}[Recovery from a $3$-drive failure]
  Since Theorem~\ref{thm:3-drive-list-decoding} yields a list of
  triples suspected of failure, we can save the relevant data and
  write new data to the suspected locations. Assuming that failed
  drives are the triple $(i,j,l)$, this triple will be repeated for
  each syndrome obtained in the course of this experiment. If one
  assume that can repeatedly generate errors from the failing drives,
  we will be able to find the failed drives by an elimination process.
\end{example}

\section{Degraded Modes}
\label{sec:degraded_mode}
When some drives are removed from a RAID array, we say that the array
is in ``degraded mode''. For example RAID~6 can operate with two drives removed,
but it loses all error detecting and correcting capabilities.
It is implied that the locations of the removed drives are known.
This modifies the recovery problem. We are interested in a maximum
likelihood algorithm, which seeks a solution with the highest
likelihood, which is equivalent to giving priority to solutions
with a smaller number of failed drives. The \emph{a priori} knowledge
that some drives are missing changes the order in which the solutions
are presented. It is assumed that parities are calculating without
the missing drives, or equivalently, the data from the missing
drives is replaced with zeros. Loosely speaking, we treat the missing
drives as having been \emph{erased}.

The algorithm based on the parity check matrix
\eqref{eqn:parity-check-matrix} creates a number of special cases,
depending on the number and role (parity or data) of the missing
drives.

\subsection{A method to handle combinations of erasures and errors}
We recall that equations~\eqref{eqn:specialized-linear-system-full-form}
and ~\eqref{eqn:specialized-linear-system} can be used to perform syndrome decoding.
We will further develop notations helpful in describing decoding
of arbitrary combinations of erasures at known locations and errors at
unknown locations.

Let $I\subseteq{1,2,\ldots,k}$ and $J\subseteq{1,2,\ldots,5}$ by any subsets
satisfying
\[ |I| + |J| \leq 5 \]
where $I\cup J$ is the set of locations of known erasures.
It is easy to see that all solutions to
the equation $H\cdot e = s$ are obtained by first solving
\begin{equation}
\label{eqn:degraded-subsystem}
  P^{J,I}e^{I} = s^{J}
\end{equation}
where $e^{I}$ is the vector obtained from the full error vector by
keeping entries in the set $I$, $P^{J,I}$ is obtained from the parity
matrix $P$ by \textbf{keeping} only columns with indices in the set
$I$ and \textbf{deleting} rows with indices in the set $J$, and
$s^{J}$ is the vector obtained from the syndrome vector $s$ by
\textbf{deleting} entries in the set $J$. In short, equations with
indices in the set $J$ are deleted. After solving equation
\eqref{eqn:degraded-subsystem} we construct $e$ by setting $e_i$ to
the $i$-th entry of $e^{I}$ for $i\in I$ and setting $e_i=0$ for
$i\notin I$, $1\leq i \leq k$. Entries $e_{k+j}$ are then uniquely
determined by the $j$-th equation of  the system
$H\cdot e = s$, which reduces to:
\[ e_{k+j} = s_j - \sum_{i\in I} p_{j,i} e_i, \qquad j\in J. \]
The matrix equation \eqref{eqn:degraded-subsystem} is a system of
\emph{non-linear} equations for the error locators $\rho_i$, and the
error values $x_i$, $i\in I$. It will be analyzed by methods of
algebra. Some of the systems are hard enough to require advanced
methods, such as Gr\"obner basis calculations and elimination
theory. Sometimes the calculations are lengthy enough to be performed
with the aid of a Computer Algebra System (CAS). In our calculations
we used the free, open source CAS Maxima \cite{Maxima}, containing a
Gr\"obner basis package written by one of the authors of this paper
(Rychlik). Let us write down the explicit form of
system~\eqref{eqn:degraded-subsystem}:
\begin{equation}
  \label{eqn:degraded-subsystem-explicit-form}
\left(\begin{array}{cccc}
        p_{j_1,i_1} & p_{j_1,i_2} & \ldots & p_{j_1,i_r} \\
        p_{j_2,i_1} & p_{j_2,i_2} & \ldots & p_{j_2,i_r} \\        
        \vdots     & \vdots    & \ddots & \vdots \\
        p_{j_t, i_1} & p_{j_t,i_2} & \ldots & p_{j_t,i_r} \\
\end{array}\right)
\cdot 
\left(\begin{array}{c}
  e_{i_1} \\ e_{i_2} \\ \vdots \\ e_{i_r}
\end{array}\right)
= \left(\begin{array}{ccc}
  s_{j_1} \\ s_{j_2} \\ \vdots \\ s_{j_t}
\end{array}\right)
\end{equation}
where $P=[p_{j,i}]$ is the parity matrix, $r=|I|$, $t=5-|J|$, and
\begin{align*}
  \{i_1,i_2,\ldots,i_r\} &= I,\\
  \{j_1,j_2,\ldots,j_t\} &= \{1,2,\ldots,5\}\setminus J.
\end{align*}
Clearly, this notation generalizes any systematic code with $5$
parities, and can be further generalized to any number of parities.
Matrix equation~\eqref{eqn:degraded-subsystem-explicit-form} can also
be written using summation notation:
\begin{equation}
  \label{eqn:degraded-subsystem-tensor-form}
  \sum_{\mu=1}^r p_{j_\nu,i_\mu} e_{i_{\mu}} = s_{j_\nu},\qquad\nu=1,2,\ldots,t.
\end{equation}
When $P$ is an algebraic function of locators $\rho_{i}$,
$i=1,2,\ldots,k$ of data, this is an algebraic system. Our matrix is
explicitly an algebraic function of locators, and thus system
\eqref{eqn:degraded-subsystem-tensor-form} is thus explicitly a
polynomial system of equations.  It should be noted that any function
on a vector space $\FF^k$ with values in $\FF$ has a polynomial
representation, and thus the method of reducing the decoding problem
to a system of polynomial equations is universally applicable to all
linear, systematic codes, and even more general classes of
codes. However, the computational complexity of the decoding algorithm
depends on the algebraic complexity of the system (this term used
loosely, as there is no rigorous, universal notion of algebraic
complexity). We recall that for the code given by \eqref{eqn:parity-check-matrix}
we have
\begin{equation}
  p_{j,i} = \begin{cases}
    \alpha_i^{j-1}, & j=1,2,3,4,\\
    \alpha_i(\alpha_i+1), & j=5.
  \end{cases}
\end{equation}
resulting in a system \eqref{eqn:degraded-subsystem-tensor-form} of
total degree at most $4$ in variables $\rho_\mu=\alpha_{i_\mu}$ (error
locators) and $x_\mu=e_{i_\nu}$ (error values), for $i=1,2,\ldots,k$.
By methods of elimination theory, solving these systems reduces to
solving polynomial equations in 1 variable. The practical
implementation of elimination theory in computational algebraic
geometry is provided by Gr\"obner basis \cite{Cox+Others/1991/Ideals}.

It should be noted that unknown location parity errors add discrete variables
$j$ in the range $1\leq j \leq 5$. They cannot be handled by algebraic
methods, or at least are inconvenient to handle. However, we may
simulate such errors by branching (as in branch-and-bound) on all
possible values.

It will be generally advantageous for given $|J|=t$ to consider
the maximum possible $r=|I|$ for which a unique solution of
\eqref{eqn:degraded-subsystem-tensor-form} exists, because
smaller $r$ are special cases obtained by setting some $x_\nu$ to
$0$ and thus are a part of the analysis for the maximum $r$.

Also, it should be noted that there is one case when
\eqref{eqn:degraded-subsystem-tensor-form} is linear, namely when all
data errors are known erasures. In this case all error locators are
known and we solve a linear system for the error values. As this
is done by the standard methods of linear algebra, it should be
considered relatively easy.

Generally, error locators for fixed $r$ are treated identically, and
thus it is beneficial to re-write
\eqref{eqn:degraded-subsystem-tensor-form} in terms of the elementary
symmetric polynomials of the error locators in order to lower the
degree of the system.

\subsection{One parity drive missing}
If $j$, $1\leq j \leq 5$, is the index of the missing parity drive
then the equation $s=H\cdot e$ is analyzed by eliminating
the variable $e_{k+j}$, which reduces to the equation
\[ P^{j}\cdot e^{j} = s^{j} \]
in which the superscript means that row $j$ of the corresponding
matrix has been erased. The above equation involves only data
drives. Two variables are associated with every drive (the locator
$\rho_i$ and the error value $x_i$, $i=1,2,\ldots,k$). Since we have
$4$ equations, we can in principle accommodate two failed data
drives. Thus we consider the equation
\[ x\,\P{u} + y\,\P{v} = s \]
using the notation introduced by \eqref{eqn:P-rho}. The details
of recovery depend on which parity drive is missing.
\subsubsection*{$j=1$} The system of equations in this case is:
\[
  \begin{pmatrix}
    u & v \\
    u^2 & v^2 \\
    u^3 & v^3 \\
    u(u+1) & v(v+1) \\
  \end{pmatrix} \cdot
  \begin{pmatrix} x \\ y \end{pmatrix} =
  \begin{pmatrix} s_2 \\ s_3 \\ s_4 \\ s_5\end{pmatrix}.
\]
In order for this system to be consistent, we have to have
$s_2+s_3=s_5$. Then the last equation is dependent and can be
discarded, resulting in:
\[
  \begin{pmatrix}
    u & v \\
    u^2 & v^2 \\
    u^3 & v^3 \\
  \end{pmatrix} \cdot
  \begin{pmatrix} x \\ y \end{pmatrix} =
  \begin{pmatrix} s_2 \\ s_3 \\ s_4\end{pmatrix}.
\]
With the aid of a CAS, we obtain the system for $u$ and $v$ alone:
\[ s_2\,u\, v + s_3\,(u + v) + s_4 = 0. \]
We use the Vieta substitution $\sigma_1=u+v$ and $\sigma_2=u\cdot v$.
We can write the above equation as a linear relationship between 
$\sigma_1$ and $\sigma_2$:
\[ s_3\,\sigma_1 + s_2\,\sigma_2 = s_4. \]
Note that if this relationship is trivial and consistent then
$s_2=s_3=s_4=0$ and $x=y=0$ is a solution. As $s_2+s_3=s_5$, also
$s_5=0$. Hence, $s$ has weight $\leq 1$ and it has a solution with no
missing data drives, and this is the solution with maximum likelihood.

If $u$ is known then
\[ (s_2\, u + s_3)v = s_3\,u+s_4 \]
Given that $(s_2\cdot u + s_3) \neq 0$, we have a unique solution
\[ v = \frac{s_3\,u+s_4}{s_2\cdot u + s_3} \]
We would like to emphasize that uniqueness does not mean existence.
The equation is inconsistent if $s_3\,u+s_4\neq 0$ and $s_2\cdot u+s_3=0$.
Therefore, non-uniqueness is only possible when
\begin{align*}
  s_2\, u + s_3 &= 0,\\
  s_3\, u + s_4 &= 0.
\end{align*}
In particular $u=s_3/s_2=s_4/s_3$, or $s_2s_4=s_3^2$. This last equation
is another necessary condition to have a solution with parity $j=1$ missing,
and a data drive missing whose locator is $u=s_3/s_2=s_4/s_3$. In particular 
$s_2, s_3, s_4\neq 0$. If this condition is satisfied then $x$ and $y$
are found by linear algebra:
\begin{align*}
  x &= \frac{s_2\,v+s_3}{u(u+v)},\\
  y &= \frac{s_2\,u+s_3}{v(u+v)}.
\end{align*}
This works when $u,v\neq 0$ and $u\neq v$, all of which can be assumed.

Hence, two-data recovery (in addition to missing parity, for a total
of \emph{three} failed drives) does not work (i.e. result in a unique
solution), unless we have another missing data drive. Hence, with only
one drive missing, we can only recover one data drive. The condition
$s_2+s_3=s_5$ serves as a parity check. We can solve the equation with
one failed data drive:
\[ x \P{u} = s \]
without using component $j=1$, which leads to
\[
  x\cdot \begin{pmatrix}
    u \\
    u^2 \\
    u^3 \\
  \end{pmatrix}
  =\begin{pmatrix} s_2 \\ s_3 \\ s_4\end{pmatrix}.
\]
In particular, if $x\neq 0$ then $s_l\neq 0$ for $l=2,3,4$ and:
\begin{align*}
  u &= \frac{s_3}{s_2}\\
  x &= \frac{s_2}{u}.
\end{align*}

\subsubsection*{$j=2$} The system is
\begin{equation*}
  \begin{pmatrix}
    1 & 1 \\
    u^2 & v^2 \\
    u^3 & v^3 \\
    u(u+1) & v(v+1) \\
  \end{pmatrix} \cdot
  \begin{pmatrix} x \\ y \end{pmatrix} =
  \begin{pmatrix} s_1 \\ s_3 \\ s_4 \\ s_5\end{pmatrix}.
\end{equation*}
Elimination using CAS produces these equations, with $\sigma_1=u+v$ and
$\sigma_2=u\cdot v$ being the symmetric polynomials:
\begin{eqnarray}
  \label{eqn:one-parity-missing-j-2}
  \sigma_{1}\,s_{5} + \sigma_{2}(s_{1}+s_{3}+s_{5})&=&s_{4}+ s_{3},\nonumber\\
  \sigma_{2}(s_{5}^2+s_{3}^2+s_{1}\,s_{3})&=&s_{4}\,s_{5}+s_{3}\,s_{4}+s_{3}^2,\\
  \sigma_{1}\,s_{3} + \sigma_{2}(\,s_{5}+s_{3})&=&s_{4}\nonumber.
\end{eqnarray}
This system has a unique solution, up to exchanging $u$ and $v$, unless
\begin{align*}
  s_5^2+s_3^2+s_1\,s_3&=0,\\
  s_4\,s_5+s_3\,s_4+s_3^2&=0
\end{align*}
Simplifying:
\begin{align*}
  (s_5+s_3)^2=s_1\,s_3,\\
  s_4(s_5+s_3)&=s_3^2
\end{align*}
We note that if $s_1=0$ then $s_5+s_3=0$ and also $s_3=0$, and
$s_5=0$. This leaves $s_4$ arbitrary. The
system~\eqref{eqn:one-parity-missing-j-2} reduces to
\begin{align*}
  0 &= s_4,\\
  0 &= 0,\\
  0 &= s_4.
\end{align*}
This implies that all syndromes are $0$, which has a solution with no errors,
which is always most likely. Hence, we may assume that $s_1\neq 0$.
If $s_3=0$ then $s_3+s_5=0$ (assuming $s_1\neq 0$), and thus $s_5=0$.
This leaves $s_1$ and $s_4$ arbitrary. The
system~\eqref{eqn:one-parity-missing-j-2} reduces to
\begin{align*}
  \sigma_2\, s_1  &= s_4,\\
  0 &= 0,\\
  0 &= s_4
\end{align*}
Therefore $s_4=0$. But then $\sigma_2 = 0$, so either $u=0$ or $v=0$. But
this is an invalid locator, so it is rejected. Hence, we assume $s_1\neq 0$
and $s_3\neq 0$. Also, $s_4\neq 0$ and $s_5+s_3\neq 0$, i.e. $s_3\neq s_5$.
Eliminating $s_5+s_3$ we obtain $(s_3^2/s_4)^2=s_1s_3$ or $s_3^3=s_1s_4^2$.
We notice that under the degeneracy condition the first and
third equation of system~\eqref{eqn:one-parity-missing-j-2} form a singular
linear system (by checking the determinant is 0). Hence, the condition indeed yields non-unique solution.

Finally, we obtain the unique solution
\begin{align*}
  \sigma_1 &= \frac{s_3(s_3+s_5)+s_1s_4}{(s_3+s_5)^2 + s_1\,s_3},\\
  \sigma_2 &= \frac{s_4(s_3+s_5) + s_3^2}{(s_3+s_5)^2 + s_1\,s_3}.
\end{align*}
given that $s_3\neq 0$ and $s_1s_3\neq (s_3+s_5)^2$.
As usual, $u$ and $v$ are the roots of the quadratic equation
\[ \zeta^2 - \sigma_1\,\zeta + \sigma_2 = 0. \]
( Again, if another missing drive is known then one of the roots is known
and the solution is unique. We will use this fact later on.)

\subsubsection*{$j=3$} The system in this case is:
\[
  \begin{pmatrix}
    1 & 1 \\
    u & v \\
    u^3 & v^3 \\
    u(u+1) & v(v+1) \\
  \end{pmatrix} \cdot
  \begin{pmatrix} x \\ y \end{pmatrix} =
  \begin{pmatrix} s_1 \\ s_2 \\s_4 \\ s_5\end{pmatrix}.
\]
Using similar methods as in other cases, we obtain the unique solution
\begin{align*}
  \sigma_1 &= \frac{s_2(s_2+s_5)+s_1\,s_4}{s_1(s_2+s_5)+s_2^2},\\
  \sigma_2 &= \frac{(s_2+s_5)^2 + s_2\,s_4}{s_1(s_2+s_5)+s_2^2}
\end{align*}
subject to the condition $s_1(s_2+s_5) + s_2^2\neq 0$.
\subsubsection*{$j=4$}
The system in this case is:
\[
  \begin{pmatrix}
    1 & 1 \\
    u & v \\
    u^2 & v^2 \\
    u(u+1) & v(v+1) \\
  \end{pmatrix} \cdot
  \begin{pmatrix} x \\ y \end{pmatrix} =
  \begin{pmatrix} s_1 \\ s_2 \\s_3 \\ s_5\end{pmatrix}.
\]
In this case, we have a solvability condition $s_1+s_2=s_5$. The last linear
equation drops out, yielding:
\[
  \begin{pmatrix}
    1 & 1 \\
    u & v \\
    u^2 & v^2 \\
  \end{pmatrix} \cdot
  \begin{pmatrix} x \\ y \end{pmatrix} =
  \begin{pmatrix} s_1 \\ s_2 \\s_3\end{pmatrix}.
\]
and the
condition on $\sigma_1$, $\sigma_2$:
\[ s_2 + \sigma_1 s_2 + s_1\sigma_2 = 0 \]
This equation is only useful assuming that we know one of the
two locators $u$ and $v$, i.e. that there is another missing disk.
\subsubsection*{$j=5$}
The system in this case is:
\[
  \begin{pmatrix}
    1 & 1 \\
    u & v \\
    u^2 & v^2 \\
    u^3 & v^3 \\    
  \end{pmatrix} \cdot
  \begin{pmatrix} x \\ y \end{pmatrix} =
  \begin{pmatrix} s_1 \\ s_2 \\s_3 \\s_4 \end{pmatrix}.
\]
The analysis yields the unique solution:
\begin{align*}
\sigma_1 &= \frac{s_1\,s_4 + s_1\,s_3}{s_1\,s_3 + s_2^2},\\
\sigma_2 &= \frac{s_2\,s_4 + s_3^2} {s_1\,s_3 + s_2^2}.
\end{align*}
This is subject to the condition: $s_1\,s_3 + s_2^2 = 0$.

\section{Error Correcting Capabilities for $4$ Failed Drives}
\label{sec:four_drives}
The method is essentially the same as for $3$ disks, so we quickly get
to the point, by establishing notation and analyzing the systems
of algebraic equations covering all cases. We note that
the inequality $Z+2\ E\leq 4$ when $Z=4$, does not allow any errors
at unknown locations. Therefore, the positions of all failed drives
are assumed to be known. The problem of finding error values is then
a linear problem, and all ingredients to solving it are now available
in the proof of Proposition~\ref{thm:distance-of-the-code-is-5}.
It should be noted that our code uses quintuple parity, which
means that with $4$ known erasures the code has still an error detecting
capability, roughly equivalent to $1$ parity check.

\subsection{One parity drive missing}
Five systems of equations are obtained from the general system
involving data error locators $(u_1,u_2,u_3)$ and data error values
$(x_1,x_2,x_3)$, by starting with the basic system
\[
  \begin{pmatrix}
    1   & 1 &1 \\
    u_1 & u_2 & u_2\\
    u_1^2 & u_2^2 & u_3^2\\
    u_1^3 & u_2^3 & u_3^3\\
    u_1(u_1+1) & u_2(u_2+1) & u_3(u_3+1)\\
  \end{pmatrix} \cdot
  \begin{pmatrix} x_1 \\ x_2 \\ x_3 \end{pmatrix} =
  \begin{pmatrix} s_1 \\ s_2 \\ s_3 \\ s_4 \\ s_5\end{pmatrix}.
\]
We consider subsystems obtained by deleting one equation, which is an
overdetermined system with $4$ equations.  We know that the
coefficient matrix after deletion of a row has a $3\times 3$ submatrix
which is non-singular, thus has rank $3$ (see proof of Proposition~\ref{thm:distance-of-the-code-is-5} and Theorem~\ref{thm:distance-of-the-code-is-5-criterion}).
Hence, the consistency condition is that the $4\times 4$ augmented matrix
has rank $3$, i.e. the determinant is $0$. Hence, for $j=1,2,3,4,5$ we
have a single polynomial which is the sufficient condition of consistency.
We thus require for a parity disk $j$ to be the failed parity that
the $4\times 4$ minors of the matrix below be singular:
\[
  \begin{pmatrix}
    1   & 1 &1  & s_1\\
    u_1 & u_2 & u_2 & s_2\\
    u_1^2 & u_2^2 & u_3^2 & s_3\\
    u_1^3 & u_2^3 & u_3^3 & s_4\\
    u_1(u_1+1) & u_2(u_2+1) & u_3(u_3+1) & s_5\\
  \end{pmatrix}
\]
One way to find the $4$ polynomials is to form a matrix by adding
a column of indeterminates $w_j$, $j=1,2,3,4,5$, and considering the determinant:
\[
  \left|\begin{matrix}
    1   & 1 &1  & s_1 & w_1\\
    u_1 & u_2 & u_3 & s_2 & w_2\\
    u_1^2 & u_2^2 & u_3^2 & s_3 & w_3\\
    u_1^3 & u_2^3 & u_3^3 & s_4 & w_4\\
    u_1(u_1+1) & u_2(u_2+1) & u_3(u_3+1) & s_5 & w_5\\
  \end{matrix}\right|.
\]
Then the polynomial equivalent to consistency with parity $j$ error is
the coefficient at $w_j$ in the above determinant. 
Moreover, the coefficients are symmetric functions of $u_1,u_2,u_3$ and
as such can be expressed in terms of elementary symmetric polynomials.
\subsubsection*{Missing parity $j=1$}
With the aid of CAS, we obtain the coefficient at $w_1$:
$$u_{1}\,u_{2}\,\left(u_{2}-u_{1}\right)\,u_{3}\,\left(u_{3}-u_{1}
 \right)\,\left(u_{3}-u_{2}\right)\,\left(s_{5}-s_{3}-s_{2}\right)$$
 Apparently, it is $0$ only if
 \[ s_2 + s_3 +s_5 = 0. \]
\subsubsection*{Missing parity $j=2$}
In this case, coefficient at $w_2$ is:
\begin{align*}
&-\left(u_{2}-u_{1}\right)\,\left(u_{3}-u_{1}\right)\,\left(u_{3}-u_{2}\right)\,\\
  &\left(u_{2}\,u_{3}\,s_{5}+u_{1}\,u_{3}\,s_{5}+u_{1}\,u
 _{2}\,s_{5}+s_{4}-u_{2}\,s_{3}\,u_{3}-u_{1}\,s_{3}\,u_{3}-s_{3}\,u_{
  3}-u_{1}\,u_{2}\,s_{3}-u_{2}\,s_{3}-u_{1}\,s_{3}\right)
\end{align*}
Only the last factor contributes a non-trivial condition (after rewriting
in terms of the elementary symmetric polynomials $\sigma_1=u_1+u_2+u_3$ and
$\sigma_2=u_1\,u_2+u_1\,u_3+u_2\,u_3$):
\[ \sigma_{2}\,\left(s_{5}+s_{3}\right)+s_{4}+\sigma_{1}\,s_{3} = 0 \]
\subsubsection*{Missing parity $j=3, 4, 5$}
In these cases, the coefficient at $w_j$ is $0$, i.e. the existence and uniqueness
is automatic.
\begin{table}[htb]
  \caption{\label{tab:fixing-3-data-1-parity} A table
    of equations related to decoding $4$ errors, where $3$ are
    data errors and $1$ is a parity error. The listing of conditions
    on the syndromes $s_j$ and data error locators $u_1$, $u_2$ and
    $u_3$ to have a solution with one parity error at position $j$, written
    in terms of the elementary symmetric polynomials
    $\sigma_1 = u_1+u_2+u_3$ and
    $\sigma_2 = u_1\,u_2+u_1\,u_3+u_2\,u_3$.}
  \begin{center}
    \begin{tabular}{||c|c||}
      \hline\hline
      $j$ & Consistency condition \\
      \hline\hline
      1 & $s_2 + s_3+s_5 = 0$ \\
      2 & $ \sigma_{2}\,\left(s_{5}+s_{3}\right)+s_{4}+\sigma_{1}\,s_{3} = 0 $\\
      3,4,5 & Empty\\
      \hline\hline
    \end{tabular}
  \end{center}
\end{table}

\subsection{Two parity drives missing}
Following the method for a single missing parity, we consider a determinant:
\[
  \left|\begin{matrix}
    1   & 1   & s_1 & w_1 & z_1\\
    u_1 & u_2  & s_2 & w_2 & z_2\\
    u_1^2 & u_2^2 & s_3 & w_3 & z_3\\
    u_1^3 & u_2^3 & s_4 & w_4 & z_4\\
    u_1(u_1+1) & u_2(u_2+1)  & s_5 & w_5 & z_5\\
  \end{matrix}\right|.
\]
A solution to the equation $H\,e=s$ exists with data error locations
given by data error locators $u_1$, $u_2$, with parity errors at
positions $j$ and $l$, iff the coefficient at $w_j\,z_l$ of the above
polynomial is $0$.  These coefficients are listed in
Table~\ref{tab:fixing-2-data-2-parity}.  It should be noted that all
consistency conditions are equations which are either linear or
quadratic in $(u_1,u_2)$ (the exception is pair $(1,4)$ which is never
satisfied; the equation is $1=0$). Therefore, if only one of the data
disks is a known erasure, these equations limit the second data disk
to at most $2$ positions, which provides a viable method to repair
RAID with $3$ erasures and $1$ failure at unknown location.

What is important about the degeneracy condition in the fourth column
of the table is that only when the syndrome vector satisfies this
condition the equation in the second column degenerates enough to
allow the possibility of more than $2$ solutions.  It is clear that
the consistency equation when treated as function of $u_1$ is a
quadratic equation, and only when all coefficients of it are $0$ the
degeneracy occurs. By comparing with
Table~\ref{tab:fixing-2-data-1-parity} we can see that the degeneracy
conditions are the consistency conditions for that case (except for
parity pair $(4,1)$ which is never consistent with $4$ errors,
$2$-parity). Hence, Table~\ref{tab:fixing-2-data-2-parity} does not
allow more than $2$ combinations of data errors, where only $1$ data
error is at a known location. If the degeneracy condition is
satisfied, there is a $2$ data, $1$ parity error consistent with the
syndromes, which is more likely.

\begin{table}[htb]
  \caption{\label{tab:fixing-2-data-2-parity} The listing of conditions
    on the syndromes $s_j$ and data error locators $u_1$ and $u_2$ to
    have a solution with parity errors at position $j$ and $l$,
    written in terms of the symmetric polynomials $\sigma_1 = u_1+u_2$
    and $\sigma_2 = u_1\,u_2$. The degeneracy condition in the last
    column is a condition for the equation expressing consistency
    to have more tan $2$ solutions.}
  \begin{center}
    \begin{tabular}{||c|c|c|c||}
      \hline\hline
      $j$ & $l$ & Consistency condition & Degeneracy condition\\
      \hline\hline
      1 & 2 &
$\sigma_{2}\,\left(s_{5}+s_{3}\right)+s_{4}+\sigma_{1}\,s_{3}$&
$s_{4}\,s_{5}+s_{3}\,s_{4}+s_{3}^2$\\
\hline

      1 & 3 &

$\sigma_{1}\,\left(s_{5}+s_{2}\right)+s_{4}+s_{2}\,\sigma_{2}$&
$s_{5}^2+s_{2}\,s_{4}+s_{2}^2$\\
\hline
      1& 4 &

             $1$&
$1$\\
\hline
      1 & 5 &

$s_{4}+\sigma_{1}\,s_{3}+s_{2}\,\sigma_{2}$&
$s_{2}\,s_{4}+s_{3}^2$\\
\hline
      2 & 3 &
              $\sigma_{2}\,s_{5}+\sigma_{1}^2\,s_{5}+\sigma_{1}\,s_{4}+s_{4}+s_{1}
              \,\sigma_{2}^2+s_{1}\,\sigma_{1}\,\sigma_{2}$&
$s_{5}^3+s_{1}\,s_{4}\,s_{5}+s_{1}\,s_{4}^2+s_{1}^2\,s_{4}$\\
\hline
      2 & 4 &
$\sigma_{1}\,\left(s_{5}+s_{3}\right)+s_{3}+s_{1}\,\sigma_{2}$&
$s_{5}^2+s_{3}^2+s_{1}\,s_{3}$\\
\hline
      2& 5&

$\sigma_{1}\,s_{4}+\sigma_{2}\,s_{3}+\sigma_{1}^2\,s_{3}+s_{1}\, \sigma_{2}^2$&
$s_{1}\,s_{4}^2+s_{3}^3$\\
\hline
      3 & 4 &

$s_{5}+s_{1}\,\sigma_{2}+\sigma_{1}\,s_{2}+s_{2}$&
$s_{1}\,s_{5}+s_{2}^2+s_{1}\,s_{2}$\\
\hline
      3 & 5 &
$s_{4}+s_{2}\,\sigma_{2}+s_{1}\,\sigma_{1}\,\sigma_{2}+\sigma_{1}^2 \,s_{2}$&
$s_{1}^2\,s_{4}+s_{2}^3$\\
\hline
      4& 5&

            $s_{3}+s_{1}\,\sigma_{2}+\sigma_{1}\,s_{2}$&
$s_{1}\,s_{3}+s_{2}^2$\\
      \hline
    \end{tabular}
  \end{center}
\end{table}  

\appendix
\section{Additional properties}
We formulate several results without a proof, which address several
specific situations which may occur when more than $2$ drives
fail. There are cases where recovery is possible.  In other cases, we
cannot recover the content of lost drives. Our results are summarized
in Table~\ref{tab:recovery-options}.

It should be noted that our primary algorithm,
Algorithm~\ref{alg:main}, searches for the error vector $e$ of minimum
weight, matching given syndrome vector $s$.  Given that the
probability of an individual disk failure is sufficiently low, this
leads to \emph{maximum likelihood decoding}, where most likely errors
are given priority over less likely errors. This results in a unique
solution if the number of failed drives is not more than $2$. If the
number reaches $3$, it may happen that there is an error vector $e$ of
weight $3$, but this vector may not be
unique. Table~\ref{tab:recovery-options} identifies situations in
which it is possible to identify most likely error vectors $e$ by an
algebraic procedure based on solving a linear system of
type~\eqref{eqn:specialized-linear-system}, but there may be many
choices of columns of $H$ which result in equally likely solutions.
The idea of a decoder producing many solutions in descending order of
likelihood is that of \emph{list decoding} \cite{Sudan_2001}. Thus,
Table~\ref{tab:recovery-options} is helpful in constructing a list
decoder. Of course, a list decoder can be based on brute force search,
which always works, but it has expensive exponential run time. 
\begin{remark}[On non-linear nature of list decoding]
It is worth noting that a list decoder must find solutions to systems
of algebraic equations. In fact, we had to solve a non-linear system
of algebraic equations in order to find the locations of the failed
drives. It is an important observation in this paper that this can be
done by solving a quadratic equation for a particular code given
by parity check matrix \eqref{eqn:parity-check-matrix}.
\end{remark}

\input recoverytable

%
\section{Intellectual property status disclosure}
An earlier version of this paper was submitted on January $24$, $2017$
to USPTO with a provisional patent application (application number:
$62/449,920$); and on January $19$, 2018 to USPTO with a full patent
application (International application number: $PCT/US18/14420$).

\TrademarkNotice



\bibliography{RaidBibliography}
\bibliographystyle{plain}

\end{document}

%% file: Table1Data2Parity.tex
\begin{table}
  \caption{\label{tab:fixing-1-data-2-parity} Systems of equations for
    recovery from a failure of $3$ drives, one data and $2$ parity at
    location $k+j$, $k+l$, $j,l=1,2,3,4,5$, $j<l$. The equations in the third column
    contain $\rho$. The fourth column contains a constraint obtained by eliminationg $\rho$
    from equations in column $2$.}
  \begin{center}
      \begin{tabular}{||c|c||l|l||}
        \hline\hline
        $j$ & $l$ & System of equations for $\rho$ & Constraints on $s$\\
        \hline\hline
\multirow{2}{*}{1} & \multirow{2}{*}{2} & $ s_{5}\,\rho+s_{4}+s_{3}$ & \multirow{2}{*}{$ s_{4}\,s_{5}+s_{3}\,s_{4}+s_{3}^2$} \\ 
                   &                    & $ s_{3}\,\rho+s_{4}   $ & \\ 
\hline
\multirow{2}{*}{1} & \multirow{2}{*}{3} & $ s_{2}\,\rho+s_{5}+s_{2} $ & \multirow{2}{*}{$s_{5}^2+s_{2}\,s_{4}+s_{2}^2 $}  \\ 
                   &                    & $ s_{5}\,\rho+s_{5}+s_{4}+s_{2} $ &  \\ 
\hline
\multirow{1}{*}{1} & \multirow{1}{*}{4} & $ (s_{5}+s_{3})\,\rho+s_{3}   $ & \multirow{1}{*}{$  s_{5}+s_{3}+s_{2}   $}  \\ 
\hline
\multirow{2}{*}{1} & \multirow{2}{*}{5} & $ s_{2}\,\rho+s_{3} $ & \multirow{2}{*}{$s_{2}\,s_{4}+s_{3}^2 $}  \\ 
                   &                    & $ s_{3}\,\rho+s_{4} $ & \\ 
\hline
\multirow{3}{*}{2} & \multirow{3}{*}{3} & $ (s_{5}+s_{1})\,\rho+s_{5}+s_{4} $ & \multirow{3}{*}{$  s_{5}^3+s_{1}\,s_{4}\,s_{5}+s_{1}\,s_{4}^2+s_{1}^2\,s_{4} $}  \\ 
                   &                    & $ s_{5}\,\rho^2+s_{4}\,\rho+s_{4} $ & \\ 
                   &                    & $ (s_{5}^2+s_{4}\,s_{5})\,\rho+s_{4}^2+s_{1}\,s_{4}   $ & \\ 
\hline
\multirow{2}{*}{2} & \multirow{2}{*}{4} & $ s_{1}\,\rho+s_{5}+s_{3} $ & \multirow{2}{*}{$  s_{5}^2+s_{3}^2+s_{1}\,s_{3}   $}  \\ 
                   &                    & $(s_{5}+s_{3})\,\rho+s_{3}   $ &\\ 
\hline
\multirow{3}{*}{2} & \multirow{3}{*}{5} & $ s_{1}\,\rho^2+s_{3} $ & \multirow{3}{*}{$  s_{1}\,s_{4}^2+s_{3}^3   $}  \\ 
                   &                    & $ s_{3}\,\rho+s_{4} $ & \\ 
                   &                    & $ s_{1}\,s_{4}\,\rho+s_{3}^2   $ & \\ 
\hline
\multirow{2}{*}{3} & \multirow{2}{*}{4} & $ s_{1}\,\rho+s_{2} $ & \multirow{2}{*}{$  s_{1}\,s_{5}+s_{2}^2+s_{1}\,s_{2}   $}  \\ 
                   &                    & $ s_{2}\,\rho+s_{5}+s_{2}   $ & \\ 
\hline
\multirow{3}{*}{3} & \multirow{3}{*}{5} & $ s_{1}\,\rho+s_{2} $ & $  s_{1}^2\,s_{4}+s_{2}^3  $  \\ 
                   &                    & $ s_{2}\,\rho^2+s_{4} $ & \\ 
                   &                    & $ s_{2}^2\,\rho+s_{1}\,s_{4}   $ & \\ 
\hline
\multirow{2}{*}{4} & \multirow{2}{*}{5} & $ s_{1}\,\rho+s_{2} $ & \multirow{2}{*}{$  s_{1}\,s_{3}+s_{2}^2   $}  \\ 
                   &                    & $ s_{2}\,\rho+s_{3} $ & \\ 
        \hline\hline
      \end{tabular}
  \end{center}
\end{table}

%% file: recoverytable.tex
\begin{table}[h]
  \caption{\label{tab:recovery-options}We describe various situations in which recovery of data may be possible.}
  \begin{tabular}{|c|c|p{4in}|c|}
    \hline\hline
    \multicolumn{2}{|c|}{Failed Drive \#} &\multirow{2}{4.5in}{\centerline{Analysis of recovery options}} & \multirow{2}{*}{Recoverable?}\\
    \cline{1-2}
    Data & Parity && \\ 
    \hline\hline
    5 & 0 & There is no way
    to recover; the fifth row of $P$ is a sum of second
    and third row, and thus the relevant matrix is a singular matrix. &No\\
    \hline
    0 & 5 & We can recover, by recomputing all parity. & Yes\\
    \hline
    1 & 4 & We can recover the data
    drive first using the non-failed parity drive, and then
    recomputing the other 4 failed drives. & Yes\\
    \hline
    \multirow{3}{*}{4} & \multirow{3}{*}{1} &  We can recover if the failed parity is the second, third, or fifth. & Yes \\
    \cline{3-4}
                &&  No way to recover if the failed parity is the first, or fourth ( determinant always zero). & No \\
    \hline
    \multirow{3}{*}{2}&\multirow{3}{*}{3}& If both of the second
    and third parity are among the three failed parities, then we
    might be able to recover (based on the determinant). &Maybe\\
    \cline{3-4}
    &&If at most one of the second and third parity are among the
    three failed parities, we can recover. & Yes\\
    \hline\hline
  \end{tabular}
\end{table}
